\documentclass[a4paper]{article}
\usepackage{a4wide}

\usepackage[utf8]{inputenc} 
\usepackage[T1]{fontenc}    
\usepackage{hyperref}       
\usepackage{url}            
\usepackage{booktabs}       
\usepackage{amsfonts}       
\usepackage{nicefrac}       
\usepackage{xcolor}         
\usepackage[round]{natbib}

\usepackage{amsmath}
\usepackage{amsthm}
\usepackage{amsfonts}
\usepackage{graphicx}
\usepackage{newtxmath}      
\usepackage[algo2e,ruled,vlined]{algorithm2e}
\usepackage{tikz}
\usepackage{todonotes}
\usepackage{caption}
\usepackage{subcaption}

\usepackage{titlesec}
\titleformat*{\subsubsection}{\large\bfseries}

\newtheorem{theorem}{Theorem}
\newtheorem{lemma}[theorem]{Lemma}
\newtheorem{coro}[theorem]{Corollary}

\newtheorem{fact}[theorem]{Fact}
\newtheorem{claim}[theorem]{Claim}
\newtheorem{restate}{Theorem}
\newtheorem{observation}[theorem]{Observation}

\def\cite{\citep}
\setcitestyle{numbers,open=[,close=]}

\DeclareMathOperator{\opt}{\textsc{Opt}}

\DeclareMathOperator{\dist}{dist}
\newcommand{\cost}{\ensuremath{\mathit{cost}}}
\newcommand{\alg}{\ensuremath{\mathcal{A}}\xspace}
\newcommand{\R}{\ensuremath{\mathbb{R}}}

\newcommand{\rhs}{\textit{rhs}}
\def\1{\vvmathbb{1}}
\def\bj{\bar\jmath}

\SetKwFor{Parallel}{for}{run in parallel}{}

\title{Learning-Augmented Dynamic Power Management with Multiple States via 
New Ski Rental Bounds}

\author{%
Antonios Antoniadis\\
{\small University of Twente}\\
{\small Enschede, The Netherlands}
\and
Christian Coester\thanks{Christian Coester is supported by the Israel Academy of Sciences and Humanities \& Council for Higher Education Excellence Fellowship Program for International Postdoctoral Researchers. Research was carried out while he was at CWI in Amsterdam and supported by NWO VICI grant 639.023.812.}\\
{\small Tel Aviv University}\\
{\small Tel Aviv, Israel}\\
\and
Marek Eliáš\thanks{Research was carried out while Marek Eliáš was at CWI in
Amsterdam and was supported by NWO GROOT project number OCENW.GROOT.2019.015
(OPTIMAL)}\\
{\small Università Bocconi}\\
{\small Milan, Italy}\\
\and
Adam Polak\thanks{Supported by SNSF project \emph{Lattice Algorithms and Integer Programming} (185030)}\\
{\small EPFL}\\
{\small Lausanne, Switzerland}\\
\and
Bertrand Simon\\
{\small IN2P3 Computing Center, CNRS}\\ {\small Villeurbanne, France}
}

\date{}
\begin{document}

\maketitle

\begin{abstract}
We study the online problem of minimizing power consumption in systems with multiple power-saving states. During idle periods of unknown lengths, an algorithm has to choose between power-saving states of different energy consumption and wake-up costs. We develop a learning-augmented online algorithm that makes decisions based on (potentially inaccurate) predicted lengths of the idle periods. The algorithm's performance is near-optimal when predictions are accurate and degrades gracefully with increasing prediction error, with a worst-case guarantee almost identical to the optimal classical online algorithm for the problem. A key ingredient in our approach is a new algorithm for the online ski rental problem in the learning augmented setting with tight dependence on the prediction error. We support our theoretical findings with experiments.
\end{abstract}

\section{Introduction}
\label{sec:intro}

Energy represents up to $70\%$ of total operating costs of modern data centers~\cite{SSG} and is one of the major quality-of-service parameters in battery-operated devices. In order to ameliorate this, contemporary CPUs are equipped with sleep states to which the processor can transition during periods of inactivity. In particular, the ACPI-standard~\cite{ACPI} specifies that each processor should possess, along with the active state $C0$ that is used for processing tasks, at least one sleep state $C1$. Modern processors generally possess more sleep states $C2,\dots$; for example, current Intel CPUs implement at least 4 such $C$-states~\cite{intel}. Apart from CPUs, such sleep states appear in many systems ranging from hard drives or mobile devices to the start-stop feature found in many cars, and are furthermore often employed when rightsizing data centers~\cite{Albers19}.

Intuitively, in a ``deeper'' sleep state, the set of switched-off components will be a superset of the corresponding set in a more shallow sleep state. This implies that the \emph{running cost} for residing in that deeper state will be lower, but the \emph{wake-up} cost to return to the active state $C0$ will be higher compared to a more shallow sleep state. In other words, there is a tradeoff between the running and the wake-up cost. During each idle period, a \emph{dynamic power management (DPM)} strategy has to decide in which state the system resides at each point in time, without a-priori knowledge about the duration of the idle period. Optimally managing these sleep states is a challenging problem due to its online nature. On the one hand, transitioning the system to a too deep state could be highly suboptimal if the idle period ends shortly after. On the other hand, spending too much idle time in a shallow state would accumulate high running costs. The impact of DPM strategies in practice has been studied for instance in data centers, where each machine may be put to a sleep mode if no request is expected. See the study of Lim et al.~\cite{lim2011} on multi-tier data centers.

The special case of $2$-state DPM systems, i.e., when there is only a single sleep state (besides the active state), is essentially equivalent to the \emph{ski rental} problem, one of the most classical problems and of central importance in the area of online optimization~\cite{Phillips-Westbrook,Irani-Karlin}. This problem is defined as follows: A person goes skiing for an unknown number of days. On every day of skiing, the person must decide whether to continue renting skis for one more day or to buy skis. Once skis are bought there will be no more cost on the following days, but the cost of buying is much higher than the cost of renting for a day. It is easy to see that this captures a single idle period of DPM with a single sleep state whose running cost is $0$: The rental cost corresponds to the running cost of the active state and the cost of buying skis corresponds to the wake-up cost; transitioning to the sleep state corresponds to buying skis. Given this equivalence, the known $2$-competitive deterministic algorithm and $e/(e-1)\approx 1.58$-competitive randomized algorithm for ski rental carry over to $2$-state DPM, and these competitive ratios are tight. In fact, it was shown by Irani et al.~\cite{IraniSG03} and Lotker et al.~\cite{LotkerPR12} that the same competitive ratios carry over even to multi-state DPM. Ski rental, also known as \emph{rent-or-buy} problem, is a fundamental problem appearing in many domains not restricted to computer hardware questions. For the AI community, this problem for example implicitly appears in expert learning with switching costs: paying the price to switch to a better expert allows to save expenses in the future.

Beyond these results for the classical online setting, \cite{IraniSG03} also gave a deterministic $e/(e-1)$-competitive algorithm for the case in which the length of the idle periods is repeatedly drawn from a fixed, and known, probability distribution. When the probability distribution is fixed but unknown they developed an algorithm that \emph{learns} the distribution over time and showed that it performs well in practice. Although it is perhaps not always reasonable to assume a fixed underlying probability distribution for the length of idle periods, real-life systems do often follow periodical patterns so that these lengths can indeed be frequently predicted with adequate accuracy, see  Chung et al.~\cite{ChungBBLM02} for a specific example. Nevertheless, it is not hard to see that blindly following such predictions can lead to arbitrarily bad performance when predictions are faulty. The field of \emph{learning-augmented} algorithms \cite{MitzenmacherV20} is concerned with algorithms that incorporate predictions in a robust way.

In this work, we introduce multi-state DPM to the learning-augmented setting. Extending ideas of \cite{IraniSG03} and \cite{LotkerPR12}, we give a reduction from multi-state DPM to ski rental that is applicable to the learning-augmented setting. 
Although ski rental has been investigated through the learning-augmented algorithms lens before \cite{PurohitSK18,WeiZ20}, earlier work has focused on the optimal trade-off between \emph{consistency} (i.e., the performance when predictions are accurate) and \emph{robustness} (i.e., the worst-case performance). To apply our reduction from DPM to ski rental, we require more refined guarantees for learning-augmented ski rental. To this end we develop a new learning-augmented algorithm for ski rental that obtains the optimal trade-off between consistency and dependence on the prediction error. Our resulting algorithm for DPM achieves a competitive ratio arbitrarily close to $1$ in case of perfect predictions and its performance degrades gracefully to a competitive ratio arbitrarily close to the optimal robustness of $e/(e-1)\approx 1.58$ as the prediction error increases.

\subsection{Formal definitions}

\paragraph{Problem definition.} In the problem of \emph{dynamic power management (DPM)}, we are given $k+1$ power states denoted by $0,1,\dots,k$,
with power consumptions $\alpha_0 > \dotsb > \alpha_k\ge 0$
and wake-up costs $\beta_0 < \dotsb < \beta_k$.
For state $0$, we have $\beta_0 = 0$ and we call this the {\em active state}.
The input is a series of \emph{idle periods} of lengths $\ell_1, \dotsc, \ell_T$
received online, i.e., the algorithm does not know the length of the
current period before it ends.
During each period, the algorithm can transition to states with lower and lower
power consumption,
paying energy cost $x \alpha_i$ for residing in state $i$ for time $x$.
If $j$ is the state at the end of the idle period, then it has to pay the wake-up cost $\beta_j$
to transition back to the active state 0.
The goal is to minimize the total cost.

In the \emph{learning-augmented} setting, the algorithm receives at the beginning of the $i$th idle period a prediction
$\tau_i\ge 0$ for the value of $\ell_i$ as additional input. We define $\eta_i := \alpha_0 |\tau_i - \ell_i|$ to be the error of the $i$th prediction, and $\eta := \sum_i^T \eta_i$ to be the total prediction error.

(Continuous-time) \emph{ski rental} is the special case of DPM with $k=1$,
$\alpha_1 = 0$ and a single idle period of some length $\ell$.
In this case, we call $\alpha := \alpha_0$ the rental cost,
$\beta := \beta_1$ the buying cost, and $\ell$ the length
of the ski season. In learning-augmented ski rental, we write the single prediction as $\tau:=\tau_1$.

\paragraph{$\pmb{(\rho,\mu)}$-competitiveness.}
Classical online algorithms are typically analyzed in terms of \emph{competitive ratio}. A (randomized) algorithm $\alg$ for an online minimization problem is said to be \emph{$\rho$-competitive} (or alternatively, obtain a \emph{competitive ratio} of $\rho$) if for any input instance,
\begin{align}\label{eq:compClassical}
\cost(\alg)\leq \rho\cdot \opt\,+\,c,
\end{align}
where $\cost(\alg)$ and $\opt$ denote the (expected) cost of \alg and the optimal cost of the instance and $c$ is a constant independent of the \emph{online} part of the input (i.e., the lengths $\ell_i$ in case of DPM). For the ski rental problem one requires $c=0$, since the trivial algorithm that buys at time $0$ has constant cost $\beta$.

In the learning-augmented setting, for $\rho\ge 1$ and $\mu\ge 0$, we say that \alg is \emph{$(\rho,\mu)$-competitive} if 
  \begin{align}
    \cost(\alg)\leq \rho\cdot \opt\, +\, \mu\cdot\eta \label{eq:rhomu}
  \end{align}
  for any instance, where $\eta$ is the prediction error. This corresponds to a competitive ratio of $\rho+\mu\frac{\eta}{\opt}$ (with $c=0$). While this could be unbounded as $\eta/OPT\to\infty$, our DPM algorithm achieves a favorable competitive ratio even in this case (see Theorem~\ref{thm:DPM}, where we take the minimum over a range of pairs $(\rho,\mu)$, including $\mu=0$).

For a $(\rho,\mu)$-competitive algorithm, $\rho$ is also called the \emph{consistency}  (i.e., competitive ratio in case of perfect
predictions)
while $\mu$ describes the dependence on the prediction error.

\subsection{Our results}

Our first result is a {\em $(\rho,\mu)$-competitive algorithm} for ski rental that achieves the optimal $\mu$ corresponding to the given $\rho$.
For $\rho \in [1, \frac{e}{e-1}]$, let
\begin{equation}
\label{eq:mu}
\mu(\rho):= \max\bigg\{ \frac{1-\rho\frac{e-1}{e}}{\ln2},\,
	\rho(1-T)e^{-T} \bigg\},
\end{equation}
where $T\in[0,1]$ is the solution to $T^2e^{-T}=1-\frac{1}{\rho}$.
Let $\tilde\rho \approx 1.16$ be the value of $\rho$ for which both terms in the maximum yield the same value. The first term dominates for $\rho > \tilde\rho$ and the second term if $\rho<\tilde\rho$.
Note that $\mu(1) = 1$ and $\mu\left(\frac{e}{e-1}\right) = 0$.
See Figure~\ref{fig:mu_r} (left) for an illustration.

\begin{theorem}
\label{thm:ski-UB}
For any $\rho\in[1,\frac{e}{e-1}]$, there is a $(\rho,\mu(\rho))$-competitive randomized
algorithm for learning-augmented ski rental, i.e., given a prediction with error $\eta$,
its expected cost is at most $\rho\opt + \mu(\rho)\cdot \eta$.
\end{theorem}

Note that $\rho<1$ is impossible for any algorithm (due to the case $\eta=0$) and $\rho>\frac{e}{e-1}$ is uninteresting since $\rho=\frac{e}{e-1}$ already achieves the best possible value of $\mu=0$.

We also prove a lower bound showing that $\mu(\rho)$ defined in \eqref{eq:mu}
is the best possible.

\begin{theorem}
\label{thm:ski-LB}
For any $\rho \in [1, \frac{e}{e-1}]$ and any (randomized) algorithm $\alg$,
there is a ski rental instance with some prediction error $\eta$
such that the expected cost of $\alg$ is at least
$\rho\opt + \mu(\rho) \eta$.
\end{theorem}

However, for most values of the prediction $\tau$
it is possible to achieve a better $\mu < \mu(\rho)$, and $\mu(\rho)$ only captures the worst case over all possible predictions $\tau$.
The proof of Theorem~\ref{thm:ski-UB} is first sketched in Section~\ref{sec:ski} for clarity, after which we provide the complete analysis while describing how to achieve the best possible $\mu$
as a function of both $\rho$ and $\tau$. Theorem~\ref{thm:ski-LB} is proved in Section~\ref{sec:proofLB}.

\begin{figure}
\includegraphics[height=4cm]{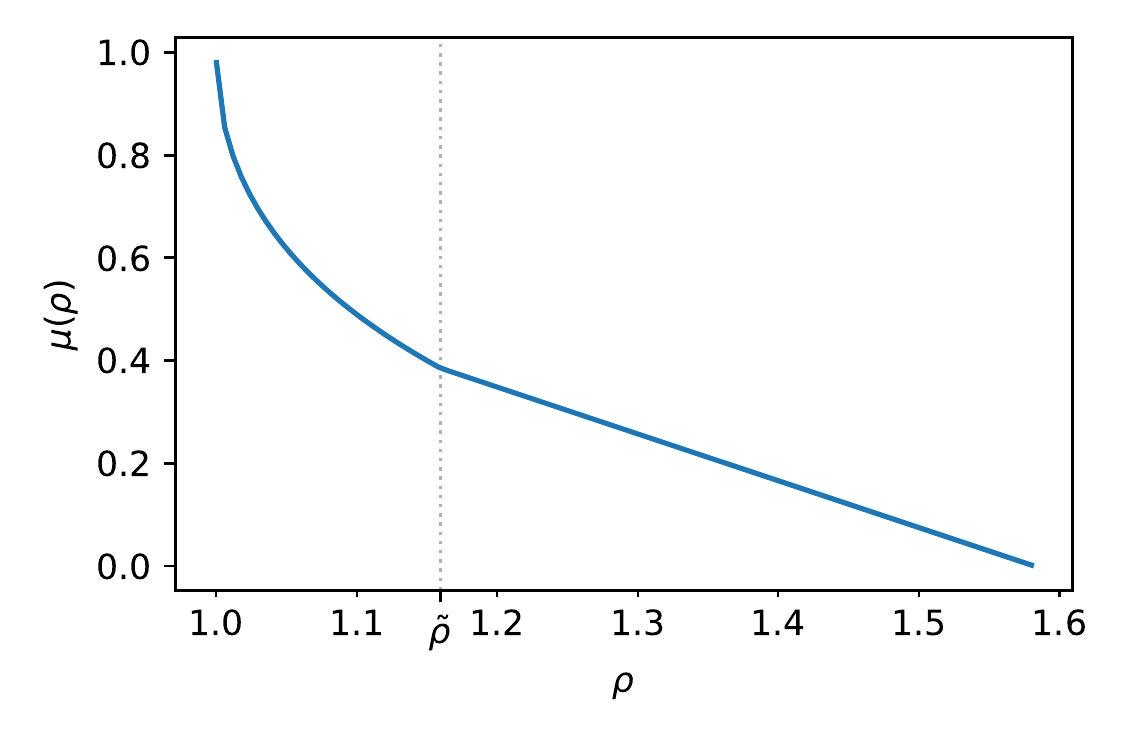}
\hfill
\includegraphics[height=4cm]{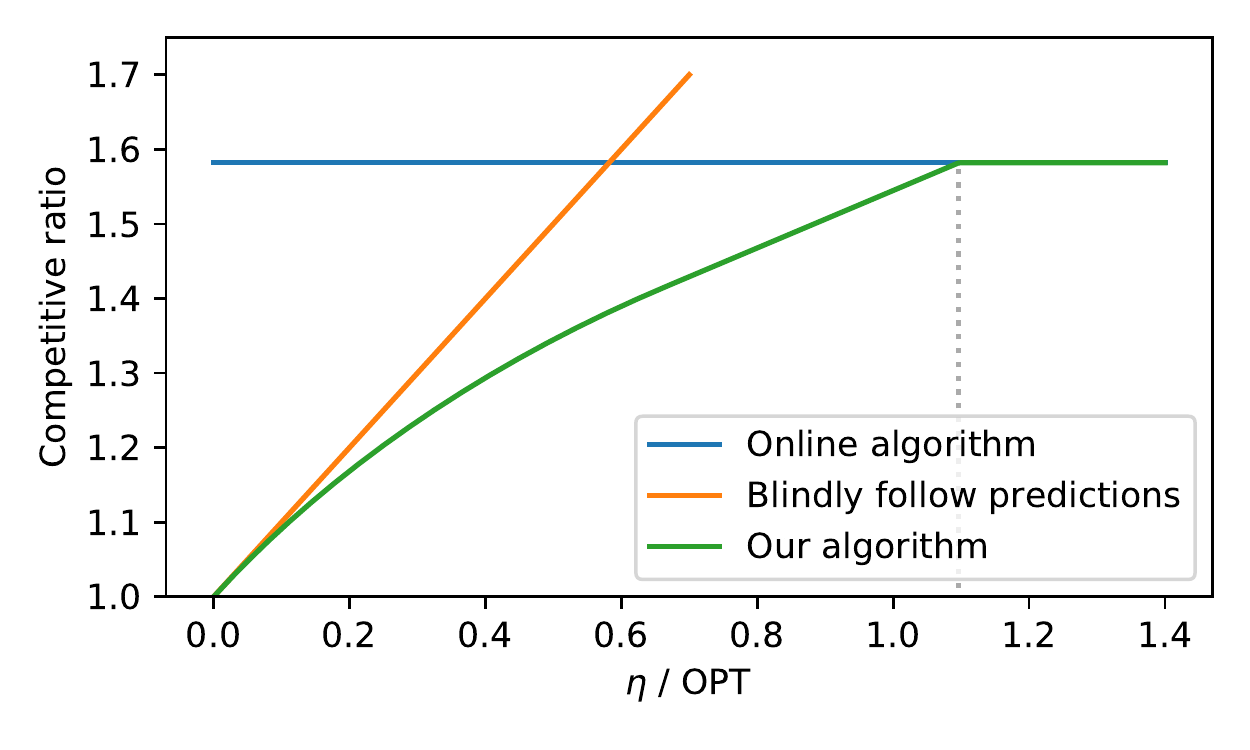}
\vspace{-2.5ex}
\caption{Illustration of $\mu(\rho)$ and of the resulting competitive ratio in function of $\eta/\opt$.}
\label{fig:mu_r}
\end{figure}

In Section~\ref{sec:DPM-1}, we give a reduction from DPM to ski rental in the learning-augmented setting, provided that the ski rental algorithm satisfies a natural monotonicity property (defined formally in Section~\ref{sec:DPM-1}):
\begin{lemma}\label{lem:redDPMSki}
	If there is a monotone $(\rho,\mu)$-competitive ski rental algorithm, then there is a $(\rho,\mu)$-competitive algorithm for DPM.
\end{lemma}
Since our ski rental algorithm is monotone, this directly yields a $(\rho,\mu(\rho))$-competitive algorithm for DPM. From the special case $(\rho,\mu)=\left(\frac{e}{e-1},0\right)$, this theorem directly implies the following result for classical DPM (without predictions),
which was first proved by Lotker et al.~\cite{LotkerPR12} for the equivalent multi-slope ski rental problem:
\begin{coro}[\cite{LotkerPR12}]
  \label{coro:dpm158}
	There is a $\frac{e}{e-1}$-competitive randomized online algorithm for DPM (without predictions).
\end{coro}
 
Using techniques from online learning, in a way similar to~\cite{AntoniadisCE0S20}, we show in Section~\ref{section:BlumBurch} how to achieve ``almost'' $(\rho,\mu(\rho))$-competitiveness simultaneously for \emph{all} $\rho$:

\begin{theorem}
\label{thm:DPM}
For any $\epsilon>0$, there is a learning-augmented algorithm $\alg$ for dynamic power management
whose expected cost can be bounded as
\[ \textstyle
\cost(\alg) \leq
(1+\epsilon) \min\big\{ \rho \opt + \mu(\rho)\cdot \eta\,\big\vert\;
	\rho\in [1, \frac{e}{e-1}]\big\}
	+ O\big(\frac{\beta_k}{\epsilon}\log \frac1\epsilon\big).
\]
\end{theorem}

The above theorem gives a competitive ratio arbitrarily close to 
$\min \{\rho + \mu(\rho)\cdot \frac{\eta}{\opt}\}$,
which is equal to $1$ if $\eta = 0$ and never greater than
$\frac{e}{e-1}$. In particular, we achieve a performance that degrades gracefully from near-optimal consistency to near-optimal robustness as the error increases.\footnote{At first glance, our consistency and robustness might seem to contradict the lower bound of \citet{WeiZ20} for ski rental. However, \cite{WeiZ20} crucially uses $c=0$ in the definition of competitiveness for ski rental.} See Figure~\ref{fig:mu_r} (right) for an illustration. 

In Section~\ref{sec:exp}, we illustrate the performance of these algorithms by simulations on synthetic datasets, where the dependence on the prediction error can be observed as expected from theoretical results.

\subsection{Related work} 
\label{sec:related_work}
\paragraph{Learning-Augmented Algorithms.}
Learning augmented algorithms have been a very active area of research since the seminal paper of Lykouris and Vassilvitskii~\cite{LykourisV18}. We direct the interested reader to a survey~\cite{MitzenmacherV20} by Mitzenmacher and Vassilvitskii, as well as \cite{AntoniadisGKK20-secretaries,secretaries-dutting, LykourisV18, AntoniadisCE0S20, Rohatgi20, AlexWei20, double-coverage, LattanziLMV20} for recent results on secretary problems, paging, $k$-server as well as scheduling problems. In the following we survey some results in the area more closely related to our work.

The ski rental problem has already been studied within the context of learning augmented algorithms.
Here, the main objective was to optimize the tradeoff between consistency and robustness (performance on perfect predictions and worst-case performance). The first results are due to Purohit et al.~\cite{PurohitSK18} who propose a deterministic and a randomized algorithm. A hyperparameter allows to choose a prescribed consistency $\rho$ and leads to a corresponding robustness. They also present a linear dependency on the error: their randomized algorithm is $(\rho,\rho)$-competitive for $\rho\ge 1$, with larger $\rho$ allowing for better robustness. Note that such a guarantee of $(\rho,\rho)$-competitiveness is not valuable in our model where we do not focus on robustness as blindly following the predictions leads to a $(1,1)$-competitive algorithm. Wei and Zhang~\cite{WeiZ20} show that the consistency / robustness tradeoff achieved by the randomized algorithm of \cite{PurohitSK18} is Pareto-optimal. Angelopoulos et al.~\cite{Angelopoulos19} propose a deterministic algorithm achieving a Pareto-optimal consistency / robustness tradeoff, but with no additional guarantee when the error is small. Interestingly, these algorithms not focusing on $(\rho,\mu)$-competitiveness are naturally monotone, so easily extend to DPM by Lemma~\ref{lem:redDPMSki}, contrarily to the tight algorithm we present in this paper. Nevertheless, experimental data (Section~\ref{sec:exp}) seem to indicate that our algorithm optimizing $(\rho,\mu)$-competitiveness for ski rental leads to better algorithms for DPM.
A variant with multiple predictions was also studied in~\cite{GollapudiP19}.

As we will see in Section~\ref{section:BlumBurch}, DPM can be cast as a problem from the class of \emph{Metrical Task Systems (MTS)}. Antoniadis et al.~\cite{AntoniadisCE0S20} gave a learning-augmented algorithm for MTS that can be interpreted as $(1,4)$-competitive within their prediction setup.

A different problem related to energy conservation is the classical online speed scaling problem, which was recently studied in the learning-augmented setting by Bamas et al.~\cite{BamasMRS20}.

\paragraph{DPM.}
The equivalence between $2$-state DPM and ski rental is mentioned in \cite{Phillips-Westbrook}. Therefore the well-known $2$-competitive deterministic and an $e/(e-1)$-competitive randomized algorithm~\cite{KarlinMMO94} for the classical ski rental problem carry over to $2$-state DPM, and these bounds are known to be tight.

Irani et al.~\cite{IraniSG03} present an extension of the $2$-competitive algorithm for two-state DPM to multi-state DPM that also achieves a competitive ratio of $2$. Furthermore they give an $e/(e-1)$-competitive algorithm for the case that the lengths of the idle periods come from a fixed probability distribution.

Lotker et al.~\cite{LotkerPR12} consider what they call \emph{multi-slope ski rental} which is equivalent to the DPM problem. Among other results, they show how to reduce a $(k+1)$-slope ski rental instance to $k$ classical ski rental instances. The reduction from DPM to ski rental presented in this paper is similar, but more general in order to also be applicable in the presence of predictions with the introduced $(\rho,\mu)$-competitiveness. They furthermore show how to compute the best possible randomized strategy for any instance of the problem.

There have been several previous approaches that try to predict the length of an idle interval (see, e.g.,~\cite{ChungBBLM02,IraniSG03}, and the survey of Benini et al.~\cite{BeniniBM00}).
However, the proposed approaches to use these predictions are not robust against a potentially high prediction error.

\citet{AugustineIS08} investigate a problem generalizing DPM where transition cost is paid for going to a deeper sleep state rather than waking up and these transition costs may be non-additive (i.e., it can be cheaper to skip states).
Albers~\cite{Albers19} studies the offline version of the problem with multiple, parallel devices and shows that it can be solved in polynomial time.

Irani et al.~\cite{IraniSG07} introduced a $2$-state problem where jobs that need to be processed have a release-time, a deadline and a required processing time. This gives further flexibility to the system to schedule the jobs and create periods of inactivity so as to maximize the energy-savings by transitioning to the sleep state. For the offline version, there is an exact polynomial-time algorithm due to Baptiste et al.~\cite{BaptisteCD12}. Recently, a $3$-approximation algorithm for the multiprocessor-case was developed~\cite{AntoniadisGKK20}.

Another related problem consists of deciding which components of a data-center should be powered on or off in order to process the current load on the set of active components (see, e.g.,~\cite{AlbersQ18}). A similar problem, where jobs have individual processing times for each machine, was studied in \cite{KhullerLS10,LiK11}. Helmbold et al.~\cite{helmbold96} considered the problem of spinning down the disk of a mobile computer when idle times are expected, which is another instance of DPM.

Several surveys cover DPM, see for example~\cite{BeniniBM00,Albers10,IraniP05}.

\section{New algorithm for ski rental}
\label{sec:ski}

For $\rho\in[1,e/(e-1)]$,
we present a $(\rho,\mu(\rho))$-competitive algorithm for (learning augmented) ski rental, proving Theorem~\ref{thm:ski-UB}. The next lemma shows that it suffices to give such an algorithm for $\alpha=\beta=1$.

\begin{lemma}
  \label{lem:unit-costs}
  An algorithm $\alg'$ that is $(\rho,\mu)$-competitive for instances of the ski rental problem with $\alpha = \beta = 1$ implies a $(\rho,\mu)$-competitive algorithm $\alg$ for arbitrary $\alpha,\beta >0$.
\end{lemma}

\begin{proof}
Given $(\rho,\mu)$-competitive algorithm $\alg'$
for unit buying and rental costs, we define
an algorithm $\alg$ for an instance with
renting cost $\alpha$ and buying cost $\beta$ as follows:
Given a prediction $\tau$, simulate \alg' with
prediction $\frac\alpha\beta \tau$. If $\alg'$ buys at time
$t'$, then \alg buys at time $t = \frac\beta\alpha t'$.

Let $c(\ell)$ and $c'(\ell')$ denote the total expected costs of $\alg$ and
$\alg'$, respectively, with corresponding lengths of the ski season
$\ell$ and $\ell'=\frac\alpha\beta\ell$.
First, note that $c(\ell) = \beta c'(\ell')$:
If $\alg'$ buys before
$\ell'$, incurring cost $1$, then $\alg$ also buys, incurring cost $\beta$.
Similarly, if $\alg'$ rents for time $t' \leq \ell'$, paying cost $t'$,
then $\alg$ rents for time $t= \frac\beta\alpha t'$, paying cost
$\alpha t = \beta t'$.
Since $\alg'$ is $(\rho,\mu)$-competitive, we have
  \begin{align*}
  \textstyle
    c(\ell) = \beta c'(\ell') \le \beta \big(\rho\min\{ \frac\alpha\beta\ell, 1\}
		+ \mu |\frac\alpha\beta \tau - \frac\alpha\beta \ell| \big)
	= \rho\opt + \mu\eta,
  \end{align*}
  where $\opt = \min\{\alpha \ell, \beta\}$ is the cost of the offline optimum
and $\eta = \alpha|\tau-\ell|$ is the prediction error of the original instance.
\end{proof}

A key difference between proving $\rho$-competitiveness in the classical online setting and $(\rho,\mu)$-competitiveness in the learning-augmented setting is the following. In the online setting without predictions, the optimal competitive ratio of $\rho=e/(e-1)$ is achieved by the greedy algorithm that at all times buys with the probability that keeps \eqref{eq:compClassical} (for $c=0$) tight assuming the skiing season ends immediately after the current time. This crucially relies on the fact that the upper bound $\rho\cdot\opt$ in \eqref{eq:compClassical} is monotone (non-decreasing) and concave as a function of the length of the ski season. Neither monotonicity (if $\tau>1$) nor concavity (regardless of $\tau$) are satisfied for the upper bound $\rho\cdot\opt+\mu\cdot\eta$ in \eqref{eq:rhomu}, which substantially complicates the description and especially the analysis of our algorithm. In particular, we will use the value of the prediction $\tau$ to determine the times when we will aim for \eqref{eq:rhomu} to be tight.

\subsection{Description of the algorithm}\label{sec:algDescription}
We next describe our randomized algorithm for instances with $\alpha=\beta=1$, which can then be used to solve arbitrary ski rental instances using Lemma~\ref{lem:unit-costs}. 
Our algorithm is fully specified by the cumulative distribution function (CDF) $F_{\tau}$ of the time when the algorithm buys skis. The algorithm then draws a $p\in [0,1]$ uniformly at random and buys at the earliest time $t\in[0,\infty)$ such that $F_{\tau}(t)\ge p$.
The CDF $F_\tau$ will depend on the given prediction $\tau\geq 0$ as well as the
fixed $\rho$ and $\mu$, which can be chosen as $\mu = \mu(\rho)$, see
Equation~\eqref{eq:mu}.
Depending on the precise value of $\tau$, it might also be possible to choose a
smaller $\mu$ as we show later in Section~\ref{sec:full-up-bound}.

\paragraph{Definition of the CDF (see Figure~\ref{fig:cdf})} We denote by $P_0$ the probability of buying at time $0$ and,
for any $t>0$, we denote by $p_t$ the probability density of buying at
time $t$, so that the probability that the algorithm buys by time $x$
can be expressed as
\[ F_{\tau}(x) = P_0 + \int_0^x p_t dt. \]
For convenience, we also specify the probability
$P_\infty = 1 - (P_0 + \int_0^\infty p_t dt)$ of never buying.

To define $P_0$ and $p_t$, we distinguish three cases depending on the value of the prediction $\tau$.
Note that we always have $0\le\mu\le1\le\rho\le\frac{e}{e-1}$. 

\begin{figure}
\centering
\scalebox{0.9}
{
\input{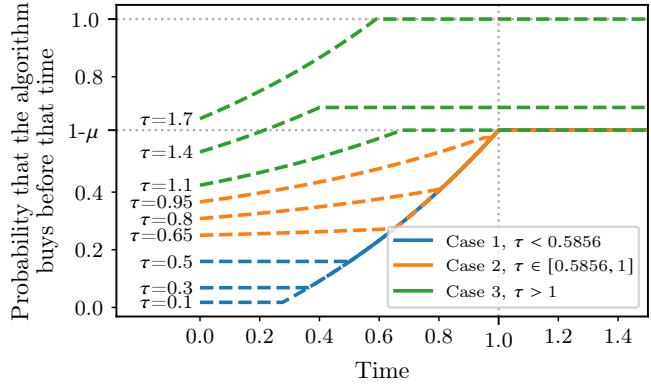}
}
\vspace{-2.5ex}
\caption{Our $(\rho,\mu)$-competitive ski rental algorithm for $\rho=\tilde\rho\approx 1.1596$ and $\mu=\mu(\tilde\rho)\approx0.3852$. The figure presents the cumulative distribution functions of the time of buying for several prediction values $\tau$. Here $\alpha=\beta=1$, i.e., at time $t=1$ buying and renting has equal costs.}
\label{fig:cdf}
\end{figure}

\paragraph{Case 1: $\boldsymbol{\mu\tau < \mu - \rho+1}$.}
We choose
\begin{align*}
P_0&=\frac{\tau(\rho-1)}{1-\tau} ,&
p_t&=\begin{cases}
\rho e^{t-1}&\text{for } t\in (b, 1]\\
0&\text{otherwise}
\end{cases}, &
P_\infty&= \min\{ \mu, 1-P_0\},
\end{align*}
where $b\in [\tau,1]$ is chosen such that
$P_0 + P_\infty + \int_{b}^1 \rho e^{t-1} dt = 1$, in order to have the sum of
probabilities equal to 1. Note that if $P_0 \geq 1-\mu$, we have $b = 1$ and
$p_t = 0$ for all $t>0$.

\paragraph{Case 2: $\boldsymbol{\mu - \rho+1 \leq \mu \tau}$ and $\boldsymbol{\tau \leq 1}$.}
We choose
\begin{align*}
P_0&= \mu\tau, &
p_t&=\begin{cases}
(\mu\tau + \rho - \mu - 1)e^t & \text{for }t\le a\\
\rho e^{t-1}&\text{for } t\in (b, 1]\\
0&\text{otherwise}
\end{cases}, &
P_\infty&= \min\{\mu, 1- P_0\},
\end{align*}
where $a \in [0,\tau]$ is chosen maximal such that
$P_0 + P_\infty+\int_{0}^a (\mu\tau+\rho-\mu-1)e^t dt \leq 1$,
and $b \in [\tau, 1]$ is chosen so that
$P_0 + P_\infty + \int_{0}^a (\mu\tau+\rho-\mu-1)e^t dt
	+ \int_{b}^1 \rho e^{t-1} dt = 1$
in order to have the sum of probabilities equal to 1.
In case $\rho=\frac{e}{e-1}$, we have $\mu=0$ and
$(\mu\tau +\rho - \mu -1)e^t = (\rho-1)e^t=\rho e^{t-1}$, recovering the classical
online algorithm of \citet{KarlinMMO94}.

\paragraph{Case 3: $\boldsymbol{\tau > 1}$.} If $\mu\tau \geq 1$, we buy at time 0.
Otherwise, we choose
\begin{align*}
P_0&=\mu \tau,  &
p_t&=\begin{cases}
(\mu \tau+\rho-\mu-1)e^t\quad&\text{if }t\le T\\
0&\text{if }t>T
\end{cases}, &
P_\infty&= \rho-\mu - (\mu \tau+\rho-\mu-1)e^T,
\end{align*}
where $T$ is the number closest to $\tau-1$ that satisfies
\begin{align}
&e^T\le \frac{\rho-\mu}{\mu \tau+\rho-\mu-1}&&\text{(equivalently $P_\infty\ge 0$)}\label{eq:T_UB}\\
&e^T\ge \frac{\rho-2\mu}{\mu \tau+\rho-\mu-1}&&\text{(equivalently $P_\infty\le
\mu$).}\label{eq:T_LB}
\end{align}
Thus, either $T=\tau-1$ if this choice satisfies both bounds, or $T$ is
at an endpoint of the feasible interval prescribed by \eqref{eq:T_UB} and
\eqref{eq:T_LB}.

\subsection{Sketch of the analysis}

We expose the main ideas of the analysis in this section.
Full analysis and the proof of Theorem~\ref{thm:ski-UB} can be found in
Section~\ref{sec:full-up-bound}.

Our algorithm is $(\rho, \mu)$-competitive if and only if for all $x\geq 0$ we have
\begin{equation}
\label{eq:compIntegralSketch}
\cost(x):=P_0+\int_0^x(1+t)\,p_t dt+\int_x^\infty x\,p_t dt + xP_\infty\le
\rho\min\{x,1\}+\mu|\tau-x|,
\end{equation}
where $\cost(x)$ denotes the expected cost of the algorithm in the case when
$\ell = x$: 
If we intend to buy at some time $t$ and $t<x$, we pay $1+t$,
otherwise we pay $x$. On the right hand side, $\min\{x,1\}$ is the optimal
cost and $|\tau-x|$ is the prediction error,
assuming $\alpha = \beta = 1$.

We first sketch the analysis for Case~2, and then discuss the differences in Case~1. These cases are relatively simple. Case~3 is far more involved and we will only sketch the ideas.

\paragraph{Case 2:}
For the algorithm to be well defined, we need to choose $\mu$
such that a suitable $b\in[\tau,1]$ exists.
For $\mu = \mu(\rho)$, this is ensured by the
inequality $\mu(\rho)\ge \frac{1-\rho\frac{e-1}{e}}{\ln 2}$ from
the definition of $\mu(\rho)$, which implies existence of such $b$ for any value
of $\tau$.
If $\tau = \ln 2$, then $\mu = \mu(\rho)$ is in fact the smallest possible,
allowing only $b = \tau$.
For other values of $\tau$, suitable $b$ exists also for smaller
values of $\mu$.
We now show that \eqref{eq:compIntegralSketch} is satisfied.

Note that \eqref{eq:compIntegralSketch} is tight for $x = 0$, with both sides equal to
$\mu \tau$. To obtain  \eqref{eq:compIntegralSketch} for all $x>0$, it suffices to show that the derivative of the left-hand side with respect to $x$ is at most the derivative of the right-hand side (where derivatives exist). For $x\in(0,\infty)\setminus\{a,b,1\}$, we have
\begin{align*}
\frac{d}{dx} \cost(x) &= (1+x) p_x + \int_x^\infty p_t dt - xp_x+ P_\infty
= p_x + \int_x^\infty p_t dt+ P_\infty.
\end{align*}
For $x \in(0,a)$ this yields
\[ \frac{d}{dx} \cost(x)
	= p_x + 1 - P_0 - (p_x - p_0)
        = 1 - \mu\tau + (\mu\tau + \rho - \mu - 1) e^0 = \rho - \mu,
\]
which is equal to the derivative of the right-hand side of
\eqref{eq:compIntegralSketch}. For $x\in(a,b)$, $\frac{d}{dx} \cost(x)$ is even smaller because $p_x$ is $0$, and the derivative of the right-hand side of \eqref{eq:compIntegralSketch} is $\rho-\mu$ or $\rho+\mu$. For $x\in (b, 1)$,
\[ \frac{d}{dx} \cost(x) = p_x + \int_x^\infty p_t dt + P_\infty
= p_x + (p_1 - p_x) + P_\infty
= \rho + P_\infty
\le \rho + \mu,
\]
which is equal to the derivative of the right-hand side of
\eqref{eq:compIntegralSketch}. Finally, for $x>1$ we have $\frac{d}{dx} \cost(x) =P_\infty\le \mu$ and the derivative of the right-hand side is also $\mu$.

\paragraph{Case 1:}
The reason we cannot define $p_t$ in the same way as in Case~2 is that $p_t$ would be negative for $t\le a$ (i.e., the algorithm would try to sell skis that it bought at time $0$, which is not allowed). We therefore choose $P_0$ such that \eqref{eq:compIntegralSketch} is tight for $x=\tau$ if we do not buy in the interval $(0,\tau]$. The remainder of the proof of \eqref{eq:compIntegralSketch} is similar to Case~2.
For $\mu = \mu(\rho)$, the existence of $b\in[\tau,1]$ follows from the
inequality $\mu\ge \rho(1-T)e^{-T}$ in the definition of $\mu(\rho)$.
Note that such $\mu$ is the smallest possible for $\tau=1-T$.

\paragraph{Case 3:} The first step in the analysis of Case 3 is to derive an inequality involving $\rho$, $\mu$, $\tau$ and $T$ that is equivalent to the algorithm being $(\rho,\mu)$-competitive. Denoting by $\mu_\tau(\rho)$ the minimal $\mu$ satisfying this inequality, it suffices to show that $\mu_\tau(\rho)\le \mu(\rho)$ for all $\tau>1$. However, the difficulty is that no closed-form expression for $\mu_\tau(\rho)$ exists. However, we are still able to show that $\tau\mapsto \mu_\tau(\rho)$ can have a local maximum only if $T=\tau-1$, and therefore $\sup_{\tau>1}\mu_\tau(\rho)$ is achieved either for $\tau\to1$ or when $T=\tau-1$. This allows us to eliminate $\tau$ from the aforementioned inequality, and we can then show that $\mu=\mu(\rho)$ satisfies the remaining inequality (with tightness occurring for $\rho\le \tilde\rho$ and $\tau=T+1$).

\subsection{Generalization and complete analysis}

\label{sec:full-up-bound}

We will show a stronger result than that of Theorem~\ref{thm:ski-UB} by considering a generalization of our algorithm to additional values of $\mu$.
For fixed $\rho\in\left[1,\frac{e}{e-1}\right]$ and a prediction $\tau\ge 0$, let $\mu_\tau(\rho)$ denote the best value of $\mu$ such that a $(\rho,\mu)$-competitive algorithm exists for instances with this particular prediction $\tau$. We will allow any $\mu\in[\mu_\tau(\rho),1]$ and show that the resulting algorithm is $(\rho,\mu)$-competitive algorithm for instances with prediction $\tau$. As this will include the value $\mu=\mu(\rho)$, Theorem~\ref{thm:ski-UB} will follow.

The benefit of this more general algorithm is that it achieves two goals at once: On the one hand, choosing $\mu=\mu_\tau(\rho)$ yields the best possible guarantee for any given prediction $\tau$. On the other hand, choosing $\mu=\mu(\rho)$ independently of $\tau$, we achieve an algorithm that satisfies the monotonicity property required by our reduction from DPM to ski rental (Section~\ref{sec:DPM-1}). Although a more specialized algorithm description only for the case $\mu=\mu_\tau(\rho)$ would suffice to prove Theorem~\ref{thm:ski-UB} (since $\mu_\tau(\rho)\le \mu(\rho)$), we remark that this algorithm would fail to be monotone, see Figure~\ref{fig:cdfs}.

\begin{figure}
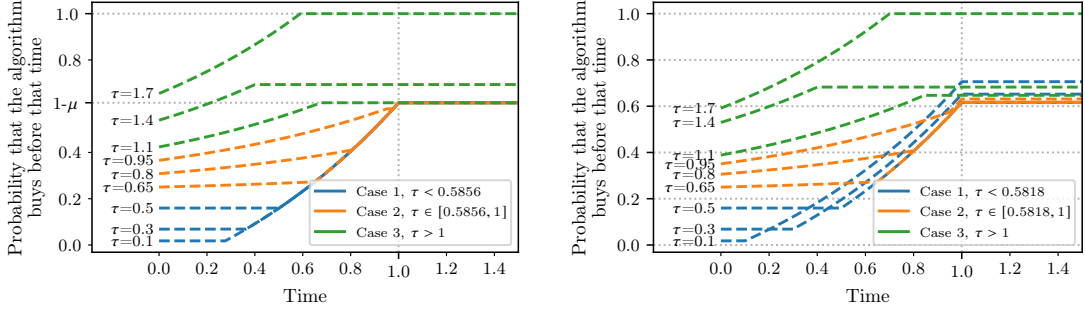

\centering
\resizebox{\linewidth}{!}
{
\input{fig/cdf.pgf}
\input{fig/cdf_mu_tau.pgf}
}
\vspace{-2.5ex}
\caption{Our $(\rho,\mu)$-competitive ski rental algorithm for $\rho=\tilde\rho\approx 1.1596$ and a) $\mu=\mu(\tilde\rho)\approx0.3852$ on the left and b) $\mu = \mu_\tau(\tilde\rho)$ on the right. The figures present the cumulative distribution functions of the time of buying for several prediction values $\tau$.}
\label{fig:cdfs}
\end{figure}

The definition of the more general algorithm is verbatim as in Section~\ref{sec:algDescription}, except we no longer fix $\mu=\mu(\rho)$ but instead consider arbitrary $\mu\in[ \mu_\tau(\rho),1]$ for $\mu_\tau(\rho)$ as defined in the next subsection.

\subsubsection{Definition of \texorpdfstring{$\boldsymbol{\mu_\tau  (\rho)}$}{}}
\label{sec:defMuTau}
For $\tau\le 1$, define
\begin{align}\label{eq:mutau}
\mu_\tau(\rho):= \begin{cases}
-\frac{\rho-1}{1- \tau}+\rho e^{ \tau-1}&\text{if }(1-\tau)e^{\tau-1}> 2-\frac{2}{\rho}\\
\frac{e^{ \tau}\left(1-\rho\frac{e-1}{e}\right)}{2-(1- \tau)e^{ \tau}}\qquad&\text{if }(1-\tau)e^{\tau-1}\le 2-\frac{2}{\rho}.
\end{cases}
\end{align}
Note that the denominators in both expressions are strictly positive, since $\tau<1$ in the first case and $e^\tau(1-\tau)\le 1$ for all $\tau\in\R$.

For $\tau>1$, we define $\mu_\tau(\rho)=\min\{\mu\mid (\mu,T)\in M_\tau(\rho)\text{ for some }T\}$, where $M_\tau(\rho)$ is the set of all pairs $(\mu,T)\in\R^2$ satisfying the following constraints:
\begin{align}
& \rho\tau+ (T-\tau)(\mu\tau+\rho-\mu-1)e^T\le \rho\label{eq:Mtau1}\\
&\rho-2\mu\le (\mu \tau+\rho-\mu-1)e^T\le \rho-\mu\label{eq:Mtau2}
\end{align}
Note that $(1/\tau,0)\in M_\tau(\rho)$, and therefore $\mu_\tau(\rho)\le 1/\tau$.

\begin{lemma}\label{lem:mutaumu}
	For each $\rho\in\left[1,\frac{e}{e-1}\right]$ and $\tau\ge 0$, we have $\mu_\tau(\rho)\le \mu(\rho)$.
\end{lemma}
\begin{proof}[Proof sketch]
	The two cases for $\tau\le 1$ are simple: If $(1-\tau)e^{\tau-1}> 2-\frac{2}{\rho}$, the upper bound can be found by substituting $T=1-\tau$ and computing the derivative with respect to $T$. If $(1-\tau)e^{\tau-1}\le 2-\frac{2}{\rho}$, the upper bound can be found by taking derivatives with respect to $\tau$. This case is tight if $\tau=\ln 2$. The proof for the case $\tau>1$ is more involved and given in Section~\ref{sec:boundingMu3}.
\end{proof}

For $\rho\in\left[1,\frac{e}{e-1}\right]$, $\tau\ge 0$ and $\mu\in[\mu_\tau(\rho),1]$, we denote by $\alg_{\rho,\mu,\tau}$ the algorithm described in Section~\ref{sec:algDescription} for ski rental instances with prediction $\tau$, except we no longer require $\mu=\mu(\rho)$. Our goal is to prove the following theorem, which combined with Lemma~\ref{lem:mutaumu} yields Theorem~\ref{thm:ski-UB}.
\begin{theorem}
For any $\rho\in\left[1,\frac{e}{e-1}\right]$, $\tau\ge 0$ and $\mu\in[\mu_\tau(\rho),1]$, algorithm $\alg_{\rho,\mu,\tau}$ is well-defined and $(\rho,\mu)$-competitive for ski rental instances with prediction $\tau$.
\end{theorem}

The following lemma shows that for $\tau\le 1$, we can essentially replace the two conditions in the definition of $\mu_\tau(\rho)$ by the conditions of Case 1 and Case 2 from the definition of our algorithm.
\begin{lemma}\label{lem:conditionCases}
	If $\tau\le 1$ and $\mu\ge\mu_\tau(\rho)$, then
	\begin{align}\label{eq:muIneq}
		\mu\ge \begin{cases}
		-\frac{\rho-1}{1- \tau}+\rho e^{ \tau-1}&\text{if }\mu\tau<\mu-\rho+1\text{ (Case 1)}\\
		\frac{e^{ \tau}\left(1-\rho\frac{e-1}{e}\right)}{2-e^{ \tau}(1- \tau)}\qquad&\text{if }\mu-\rho+1\le \mu\tau\text{ (Case 2)}.
		\end{cases}
	\end{align}
\end{lemma}
\begin{proof}
	
	\textit{Case 1} ($\mu\tau<\mu-\rho+1$): Suppose that \eqref{eq:muIneq} fails. Then
	\begin{align*}
		\rho-1< \mu(1-\tau)\le -\rho+1+\rho(1-\tau) e^{\tau-1},
	\end{align*}
	where the first inequality uses the condition of Case 1 and the second inequality uses $\tau\le 1$ and the assumption that \eqref{eq:muIneq} fails. Rearranging, we get
	\begin{align*}
		2-\frac{2}{\rho}<(1-\tau) e^{\tau-1},
	\end{align*}
	but this is precisely the first condition in \eqref{eq:mutau}. Thus, the premise $\mu\ge \mu_\tau(\rho)$ implies inequality \eqref{eq:muIneq}.
	
	\textit{Case 2} ($\mu-\rho+1\le \mu\tau$): If $(1-\tau)e^{\tau-1}\le 2-\frac{2}{\rho}$, then \eqref{eq:muIneq} follows immediately from $\mu\ge \mu_\tau(\rho)$. So suppose $(1-\tau)e^{\tau-1}> 2-\frac{2}{\rho}$. Our goal is to derive a contradiction. We have
	\begin{align*}
	\rho-1\ge \mu(1-\tau)\ge \mu_\tau(\rho)(1-\tau)= -\rho+1+\rho(1-\tau) e^{ \tau-1},
	\end{align*}
	where the first inequality uses the condition of Case 2, the second inequality uses $\tau\le 1$ and $\mu\ge \mu_\tau(\rho)$, and the equation uses our assumption $(1-\tau)e^{\tau-1}> 2-\frac{2}{\rho}$ and the definition of $\mu_\tau(\rho)$. Rearranging, we get
	\begin{align*}
	2-\frac{2}{\rho}\ge (1-\tau) e^{ \tau-1},
	\end{align*}
	yielding the desired contradiction.
\end{proof}

\subsubsection{Condition for competitiveness}
For given $\tau$, the algorithm achieves the desired $(\rho, \mu)$-competitiveness if and only if for all $x\geq 0$ we have
\begin{equation}
\label{eq:compIntegral}
\cost(x):=P_0+\int_0^x(1+t)\,p_t dt+\int_x^\infty x\,p_t dt + xP_\infty\le
\rho\min\{x,1\}+\mu|\tau-x|,
\end{equation}
where $\cost(x)$ denotes the expected cost of the algorithm in the case when
$\ell = x$: 
If we intend to buy at some time $t$ and $t<x$, we pay $1+t$,
otherwise we pay $x$. On the right hand side, $\min\{x,1\}$ is the optimal
cost and $|\tau-x|$ is the prediction error,
assuming $\alpha = \beta = 1$. We denote by $\rhs(x):=\rho\min\{x,1\}+\mu|\tau-x|$ the right hand side of \eqref{eq:compIntegral}.

The following identity will be useful. It holds for almost all $x>0$ (excluding only those $x$ where $p_x$ is discontinuous):
\begin{align}
\frac{d}{dx}\cost(x)&=(1+x)p_x + \int_x^\infty p_t d_t + P_\infty - xp_x\nonumber\\
&= p_x + \int_x^\infty p_t d_t + P_\infty\label{eq:cost'End}\\
&= p_x+1 -P_0-\int_0^x p_t d_t, \label{eq:cost'Start}
\end{align}
where the last equation uses the fact that probabilities add up to $1$.

\subsubsection{Case 1 \texorpdfstring{($\boldsymbol{\mu\tau< \mu-\rho+1}$)}{}}
Note that $\tau<1$ in this case.

\paragraph{Well-definedness.} We need to argue that $P_0\le 1$ and that $b\in[\tau,1]$ as stated in the algorithm description exists. By the condition of Case 1, we have $\rho-1\le \mu(1-\tau)$ and therefore $P_0=\frac{\tau(\rho-1)}{1-\tau}\le \mu\tau\le  1$, where the last inequality uses $\mu\le 1$ and $\tau<1$.

Regarding existence of $b$, if $P_\infty=1-P_0$ then clearly $b=1$. Otherwise, we choose
\begin{align*}
b=1+\ln\frac{\mu+\frac{\rho-1}{1-\tau}}{\rho}.
\end{align*}
Then
\begin{align*}
P_0+P_\infty+\int_b^1\rho e^{t-1} dt&= \frac{\tau(\rho-1)}{1-\tau}+\mu+\rho-\rho e^{b-1}\\
&= \frac{\tau(\rho-1)}{1-\tau}+\rho-\frac{\rho-1}{1-\tau}\\
&= 1,
\end{align*}
as required by the definition of $b$. Note that $b\le 1$ since otherwise the integral would be negative, and hence the left hand side would be less than $1$ since $P_\infty\le 1-P_0$. Moreover, by Lemma~\ref{lem:conditionCases} we have $b\ge \tau$. So indeed $b\in[\tau,1]$, i.e., the algorithm is well-defined in Case 1.

\paragraph{Competitiveness.} We need to show that \eqref{eq:compIntegral} is satisfied. First note that  
\begin{align*}
 \cost(\tau) = P_0 + \int_\tau^\infty \tau p_t dt + \tau P_\infty = P_0 + \tau(1-P_0)
= (1-\tau) \frac{\tau(\rho-1)}{1-\tau} + \tau = \rho \tau,
\end{align*}
making \eqref{eq:compIntegral} tight for $x=\tau$. To conclude \eqref{eq:compIntegral} for all $x$, it suffices to show the following equivalence for all $x\in(0,\tau)\cup(\tau,b)\cup(b,1)\cup(1,\infty)$:
\begin{align}
\frac{d}{dx}\cost(x)\le \frac{d}{dx}\rhs(x)\iff x\ge \tau.\label{eq:equivDeriv}
\end{align}
For $x<b$, from \eqref{eq:cost'Start} we get
\begin{align*}
\frac{d}{dx}\cost(x)&= 1-P_0=1- \frac{\tau(\rho-1)}{1-\tau}=\frac{1-\tau\rho}{1-\tau}.
\end{align*}
If $x<\tau$, using the condition of Case 1 we conclude $\frac{d}{dx}\cost(x)>\rho-\mu=\frac{d}{dx}\rhs(x)$.  For $x\in(\tau,b)$, \eqref{eq:equivDeriv} follows from
\begin{align*}
\frac{d}{dx}\cost(x)&= \frac{1-\tau\rho}{1-\tau}\le \rho+\mu=\frac{d}{dx}\rhs(x).
\end{align*}
For $x\in(b,1)$, \eqref{eq:cost'End} shows
\begin{align*}
\frac{d}{dx} \cost(x) = p_x + \int_x^\infty p_t dt + P_\infty
= p_x + (p_1 - p_x) + P_\infty
= \rho + P_\infty
\le \rho + \mu\le \frac{d}{dx} \rhs(x).
\end{align*}
Finally, for $x>1$ we have $\frac{d}{dx} \cost(x) =P_\infty\le \mu=\frac{d}{dx} \rhs(x)$.

We note the following observation, which will be useful later when we prove the lower bound.
\begin{observation}\label{obs:tight1}
	In Case 1, if $b=\tau$ and $P_\infty=\mu$, then \eqref{eq:compIntegral} is tight for all $x\ge \tau$.
\end{observation}

\subsubsection{Case 2 \texorpdfstring{($\boldsymbol{\mu-\rho+1\le \mu\tau}$ and $\boldsymbol{\tau\le 1}$)}{}}
\paragraph{Well-definedness.} The only non-trivial part is to argue about the existence of $b$. If $a<\tau$ or $P_\infty=1-P_0$, then $b=1$. Otherwise, $a=\tau$ and $P_\infty=\mu$ and we choose
\begin{align*}
b=1+\ln\frac{\mu\left(2-(1-\tau)e^\tau\right)+e^\tau(\rho-1)}{\rho}.
\end{align*}
Then
\begin{align*}
	P_0 + &P_\infty + \int_{0}^a (\mu\tau+\rho-\mu-1)e^t dt
	+ \int_{b}^1 \rho e^{t-1} dt \\
	&= \mu\tau+\mu+(\mu\tau+\rho-\mu-1)(e^\tau-1)+\rho-\rho e^{b-1}\\
	&= 2\mu+1+(\mu\tau+\rho-\mu-1)e^\tau-\mu\left(2-(1-\tau)e^\tau\right)-e^\tau(\rho-1)\\
	&= 1
\end{align*}
as required by the definition of $b$. We have $b\le 1$ since otherwise the left hand side would be less than $1$ (by definition of $a$). Moreover, by Lemma~\ref{lem:conditionCases},
\begin{align*}
b\ge 1+\ln\frac{e^{ \tau}\left(1-\rho\frac{e-1}{e}\right)+e^\tau(\rho-1)}{\rho}= \tau.
\end{align*}
So $b\in[\tau,1]$, as required.

\paragraph{Competitiveness.} Note that \eqref{eq:compIntegral} is tight for $x = 0$, with both sides equal to
$\mu \tau$. To obtain  \eqref{eq:compIntegral} for all $x>0$, it suffices to show that $\frac{d}{dx}\cost(x)\le \frac{d}{dx}\rhs(x)$ whenever both derivatives exist.

For $x \in(0,a)$, from \eqref{eq:cost'Start} we get
\[ \frac{d}{dx} \cost(x)
= p_x + 1 - P_0 - \int_0^xp_t dt
= 1 - \mu\tau + (\mu\tau + \rho - \mu - 1) e^0 = \rho - \mu=\frac{d}{dx}\rhs(x).
\]
For $x\in(a,b)$, $\frac{d}{dx} \cost(x)$ is even smaller because $p_x$ is $0$, and $\frac{d}{dx}\rhs(x)\in\{\rho-\mu,\rho+\mu\}$. For $x\in (b, 1)\cup(1,\infty)$, the proof of $\frac{d}{dx}\cost(x)\le \frac{d}{dx}\rhs(x)$ is identical to Case 1.

\begin{observation}\label{obs:tight2}
	In Case 2, if $a=b=\tau$ and $P_\infty=\mu$, then \eqref{eq:compIntegral} is tight for all $x\ge 0$.
\end{observation}

\subsubsection{Case 3 \texorpdfstring{($\boldsymbol{\tau>1}$)}{}}\label{sec:Case3Comp}
The case $\mu\tau\ge 1$ is trivial, so we assume throughout this section that $\mu<1/\tau$.
\paragraph{Well-definedness.} The existence of $T$ follows from $\mu<1/\tau<1$. We only need to check that probabilities add up to $1$:
\begin{align*}
P_0+P_\infty + \int_0^\infty p_t d_t=\mu \tau + \rho-\mu - (\mu \tau+\rho-\mu-1)e^T + (\mu \tau+\rho-\mu-1)(e^T-1)=1,
\end{align*}
where the first equation uses positivity of $T$:
\begin{fact}\label{fact:T0}
	$T> 0$.
\end{fact}
\begin{proof}
	We either have $T\ge \tau-1>0$ or the upper bound \eqref{eq:T_UB} is tight. In the latter case, $T=\ln\frac{\rho-\mu}{\mu \tau+\rho-\mu-1}> \ln 1=0$ since $\mu \tau< 1$.
\end{proof}

\paragraph{Competitiveness.}

For $\tau>1$ and $\mu\in[\mu_\tau(\rho),1/\tau)$, denote by $T(\tau,\mu)$ the corresponding value of $T$ chosen by the algorithm. We may drop $\tau$ in the notation when it is clear from the context.
\begin{lemma}\label{lem:TmuMon}
$T(\tau,\mu)$ is a non-increasing function of $\mu$. 
\end{lemma}
\begin{proof}
	It suffices to show that $e^{T(\mu)}$ is non-increasing in $\mu$. By definition, $e^{T(\mu)}$ is the projection of $e^{\tau-1}$ onto the interval
	\begin{align*}
	\left[\frac{\rho-2\mu}{\mu \tau+\rho-\mu-1},\frac{\rho-\mu}{\mu \tau+\rho-\mu-1}\right].
	\end{align*}
	Since both endpoints of the interval are decreasing functions of $\mu$ (since $\tau>1$), it follows that $T(\mu)$ is non-increasing in $\mu$.
\end{proof}

\begin{lemma}\label{lem:Tmuletau}
	$T(\tau,\mu)< \tau$ for every $\mu\in[\mu_\tau(\rho),1/\tau)$. Moreover, if $\mu=\mu_\tau(\rho)$, then $(\mu,T(\tau,\mu))\in M_\tau(\rho)$.
\end{lemma}
\begin{proof}
	 If the second statement holds, then since $\tau>1$, constraint~\eqref{eq:Mtau1} implies that $T(\mu_\tau(\rho))<\tau$. By Lemma~\ref{lem:TmuMon}, this implies the first statement. It remains to show the second statement.
	 
	 Let $\mu=\mu_\tau(\rho)$ and observe that there exists $T$ such that $(\mu,T)\in M_\tau(\rho)$.

By taking derivatives, one can see that the left-hand side of \eqref{eq:Mtau1} is concave in $T$ for $T<\tau-2$ and convex in $T$ for $T>\tau-2$, with a local minimum at $T=\tau-1$. Since constraint~\eqref{eq:Mtau1} is violated for $T\to-\infty$, but there exists $T$ such that $(\mu,T)\in M_\tau(\rho)$, we conclude that $T=\tau-1$ is a global minimum of the left-hand side of \eqref{eq:Mtau1} and $(\mu,T(\mu))\in M_\tau(\rho)$ by definition of $T(\mu)$.
\end{proof}

\begin{lemma}\label{lem:case3Equiv}
	For every fixed $\tau>1$ and $\mu\in[\mu_\tau(\rho),1/\tau)$, the following statements are equivalent:
	\begin{enumerate}
		\item The algorithm is $(\rho,\mu)$-competitive for instances with prediction $\tau$.
		\item $\cost(\tau)\le \rho$.
		\item Constraint \eqref{eq:Mtau1} is satisfied for $T=T(\tau,\mu)$.
	\end{enumerate}
	In the positive case, $T\le 1$.
\end{lemma}

\begin{proof}
	The implication from the first to the second statement is immediate.
	
	Using \eqref{eq:cost'Start} and \eqref{eq:cost'End}, we observe that the cost is a piecewise linear function of $x$:
	\begin{align}
	\frac{d}{dx}\cost(x)&=\left\lbrace
	\array{@{}ll@{\quad}l@{}}
	1-\mu \tau+(\mu \tau+\rho-\mu-1)\left[e^x-\int_0^x e^t dt\right]&= \rho-\mu\quad&\text{if }x< T\\
	&\phantom{{}=} P_\infty &\text{if }x> T
	\endarray \right.\label{eq:cost'}
	\end{align}
	Since $T<\tau$ by Lemma~\ref{lem:Tmuletau}, this yields
	\begin{align*}
	\cost(\tau)&=\cost(0)+(\rho-\mu)T+(\tau-T)P_\infty\\
	&=\mu \tau + (\rho-\mu)T + (\tau-T)\left[\rho-\mu - (\mu \tau+\rho-\mu-1)e^T\right]\\
	&=\rho\tau + (\tau-T)(\mu \tau+\rho-\mu-1)e^T
	\end{align*}
	
	Thus, $\cost(\tau)$ is precisely the left-hand side of~\eqref{eq:Mtau1}, which shows the equivalence of the second and third statement.
	
	It remains to show that the second and third statement imply  $T\le 1$ and that \eqref{eq:compIntegral} holds for \emph{all} $x$. 
	
	We must have $T\le 1$ since otherwise we get a contradiction to the second statement via
	\begin{align*}
	\cost(\tau)> \cost(1)=\mu \tau+\rho-\mu \ge\rho.
	\end{align*}
	For $x=0$, \eqref{eq:compIntegral} is satisfied with equality. For $x\in(0,T)$, by \eqref{eq:cost'} and since $T\le 1<\tau$ we have $\frac{d}{dx}\cost(x)=\frac{d}{dx}\rhs(x)$, so \eqref{eq:compIntegral} is tightly satisfied for all $x\in[0,T]$. By assumption we also satisfy it for $x=\tau$. Since $\cost(x)$ is linear in $x$ for $x>T$ and $\rhs(x)$ is concave for $x\in[T,\tau]$ (linearly increasing for $x\in[T,1]$ and linearly decreasing for $x\in[1,\tau]$), this shows that \eqref{eq:compIntegral} also holds for $x\in[T,\tau]$. For $x>\tau$, we have $\frac{d}{dx}\cost(x)=P_\infty\le\mu=\frac{d}{dx}\rhs(x)$ (where the inequality uses \eqref{eq:T_LB}), so \eqref{eq:compIntegral} holds there as well. See Figure~\ref{fig:cost_rhs} for an illustration of these functions.

\begin{figure}
\centering
{
\input{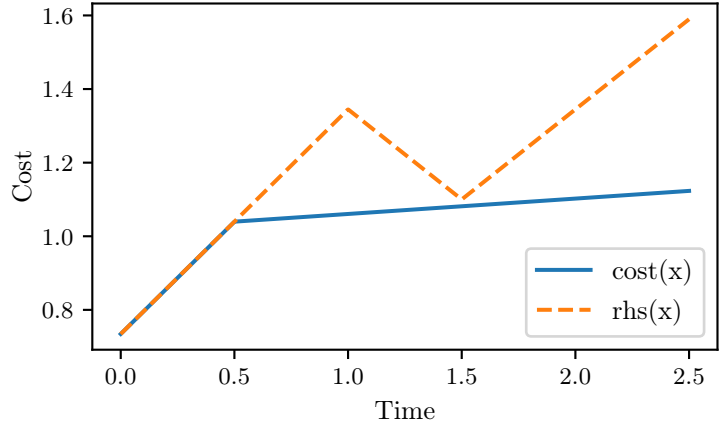}
}
\vspace{-2.5ex}
\caption{Evolution of $\cost$ and $\rhs$ in function of the time, for $\rho=1.1$, $\mu = \mu(\rho)$, $\tau=1.5$, so $T=0.5$.}
\label{fig:cost_rhs}
\end{figure}

\end{proof}

The following observation is immediate from the proof of Lemma~\ref{lem:case3Equiv}.
\begin{observation}\label{obs:tight3}
	Let $\tau>1$ and $\mu\in[\mu_\tau(\rho),1/\tau)$. Then \eqref{eq:compIntegral} is tight for all $x\in[0,T]$. If $P_\infty=\mu$ and \eqref{eq:compIntegral} is tight for $x=\tau$, then it is tight for all $x\ge\tau$.
\end{observation}

For $\tau>1$ and $\mu\in[\mu_\tau(\rho),1/\tau)$, we may write $\cost(\tau,\mu)$ instead of $\cost(\tau)$ to emphasize the additional dependence on $\mu$. By (the proof of) Lemma~\ref{lem:case3Equiv}, $\cost(\tau,\mu)$ is equal to the left-hand side of \eqref{eq:Mtau1}, i.e.,
\begin{align}
\cost(\tau,\mu)=\rho\tau+ (T(\tau,\mu)-\tau)(\mu\tau+\rho-\mu-1)e^{T(\tau,\mu)}\label{eq:costtaumu}
\end{align}

It remains to argue that the equivalent statements of Lemma~\ref{lem:case3Equiv} are indeed true for all $\mu\in[\mu_\tau(\rho),1/\tau)$. For $\mu=\mu_\tau(\rho)$, this follows from $(\mu,T(\mu)\in M_\tau(\rho)$ as shown in Lemma~\ref{lem:Tmuletau}. The following lemma implies that this remains true for larger $\mu$.

\begin{lemma}\label{lem:costDecr}
For every $\tau>1$, the function $\mu\mapsto\cost(\tau,\mu)$ is decreasing.
\end{lemma}
\begin{proof}	
	Consider first the effect of increasing $\mu$ within some interval $I$ such that one of the bounds \eqref{eq:T_UB} or \eqref{eq:T_LB} is tight for all $\mu\in I$. This means that the total probability of buying before time $\tau$ is the same for all $\mu\in I$, and thus the buying cost contribution to $\cost(\tau,\mu)$ is the same for all $\mu\in I$. However, increasing $\mu$ means that the algorithm buys more aggressively (i.e., earlier), which leads to a decrease in the expected rent cost. Thus, increasing $\mu$ in $I$ leads to a decrease of $\cost(\tau,\mu)$.
	
	Now, consider the effect of increasing $\mu$ within an interval $I$ where neither \eqref{eq:T_UB} nor \eqref{eq:T_LB} is tight. Then $T(\tau,\mu)=\tau-1$ is constant for all $\mu\in I$. In this case, due to \eqref{eq:costtaumu} and since $\tau>1$, $\cost(\tau,\mu)$ is again decreasing in $\mu$.
\end{proof}

\subsubsection{Bounding \texorpdfstring{$\boldsymbol{\mu_\tau(\rho)}$}{} in Case 3}\label{sec:boundingMu3}

The goal of this section is to prove the following Lemma:
\begin{lemma}\label{lem:muBoundCase3}
	For any $\rho\in\left[1,\frac{e}{e-1}\right]$ and $\tau>1$,
	\begin{align*}
	\mu_\tau(\rho)\le\max\left\{\frac{\rho+e-\rho e}{2},\rho(1-T)e^{-T}\right\},
	\end{align*}
	where $T\in[0,1]$ is the solution to $T^2e^{-T}=1-\frac{1}{\rho}$.
\end{lemma}

Since $2>e\ln(2)$, this will complete the proof of Lemma~\ref{lem:mutaumu}.

The difficulty in this case, compared to upper bounding $\mu_\tau(\rho)$ in Cases 1 and 2, arises from the fact that there exists no closed-form expression for $\mu_\tau(\rho)$.

The lemma is trivial for $\rho=1$ since $\mu_\tau(1)\le 1/\tau<1$ and due to the solution $T=0$. Therefore, assume throughout the remainder of this section that $\rho>1$. Since we view $\rho$ as fixed, we may write $\mu_\tau$ instead of $\mu_\tau(\rho)$.

By the next Lemma, the interval $[\mu_\tau,1/\tau)$ is indeed non-empty and constraint~\eqref{eq:Mtau1} is tight for $\mu=\mu_T$ and $T=T(\tau,\mu)$.

\begin{lemma}\label{lem:rhoTight}
	For all $\tau>1$, we have $\mu_\tau<1/\tau$ and $\cost(\tau,\mu_\tau)=\rho$.
\end{lemma}
\begin{proof}
	The inequality $\mu_\tau<1/\tau$ can be seen by noting that for $\mu<1/\tau$, we have $(\mu,0)\in M_\tau(\rho)$ provided that $\mu$ is close enough to $1/\tau$. Recall that we assumed $\rho>1$.
	
	For the other statement, the direction ``$\le$'' is immediate from the fact that $\cost(\tau,\mu)$ is the left-hand side of \eqref{eq:Mtau1}. For the other direction, consider first the case $\mu_\tau=0$. Then $(\mu\tau+\rho-\mu-1)e^{T(\tau,\mu)}=\rho$ by definition of $T(\tau,\rho)$, so $\cost(\tau,\mu)=\rho T(\tau,\rho)=\rho\ln\frac{\rho}{\rho-1}\ge\rho$, where the last inequality uses $\rho\le \frac{e}{e-1}$.
	
	Consider now $\mu_\tau>0$. If constraint~\eqref{eq:Mtau1} were strict for $T=T(\tau,\mu_\tau)$, then by continuity of the constraints~\eqref{eq:Mtau1} and \eqref{eq:Mtau2} one could find a pair $(\mu,T)\in M_\tau(\rho)$ with $\mu<\mu_\tau$, contradicting the minimality of $\mu_\tau$.
\end{proof}

As we will see in the next lemma, $\tau\mapsto\mu_\tau$ can have a maximum for $\tau>1$ only if $T(\tau,\mu_\tau)=\tau-1$. In fact, there can be at most one local maximum for $\tau>1$ (cf. Figure~\ref{fig:1localMax}).

\begin{lemma}\label{lem:localMax}
	Suppose $\tau\mapsto \mu_\tau$ has a maximum at $\tau^*>1$, and write $\mu^*:=\mu_{\tau^*}=\max_\tau\mu_\tau$. Then $T(\tau^*,\mu^*)=\tau^*-1$.
\end{lemma}
\begin{figure}
	\centering\includegraphics[width=.5\textwidth]{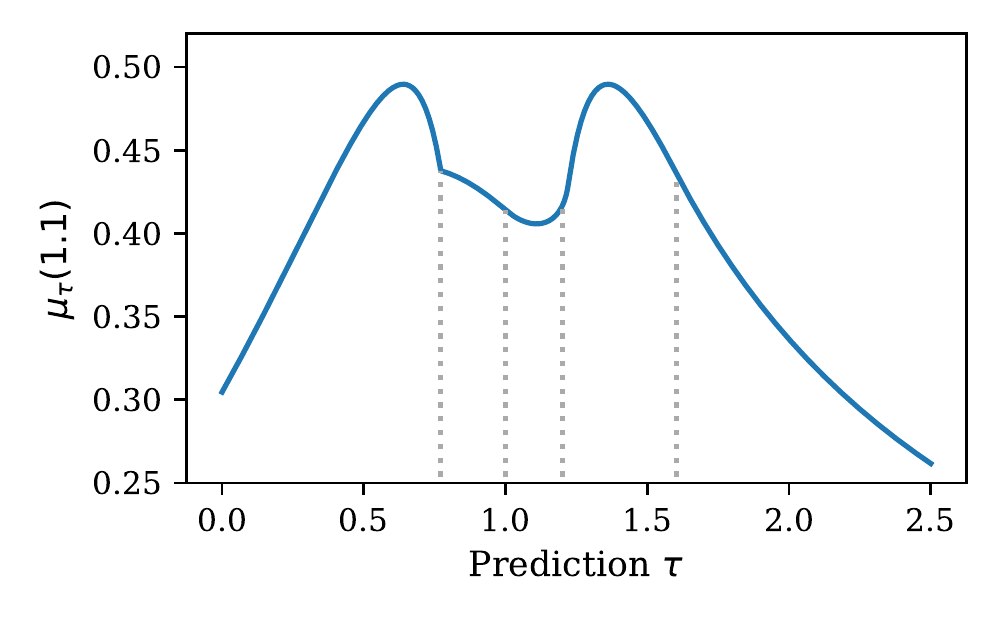}
	\caption{$\mu_\tau(\rho)$ as a function of $\tau$ for $\rho=1.1$. The vertical lines separate the five cases of interest. From left to right, we distinguish Case 1, Case 2, and three subcases of Case 3 ($\tau>1$) where: $P_\infty = \mu$; $0<P_\infty<\mu$; $P_\infty = 0$.}\label{fig:1localMax}
\end{figure}
\begin{proof}	
	Suppose towards a contradiction that $T(\tau^*,\mu^*)\ne \tau^*-1$. We will show that there exists $\tau\ne \tau^*$ such that $\cost(\tau,\mu^*)>\cost(\tau^*,\mu^*)$. Combined with the fact that $\cost(\tau^*,\mu^*)=\rho=\cost(\tau,\mu_\tau)$ (Lemma~\ref{lem:rhoTight}), we get $\cost(\tau,\mu^*)>\cost(\tau,\mu_\tau)$. But then Lemma~\ref{lem:costDecr} implies that $\mu^*<\mu_\tau$, contradicting the choice of $\mu^*$.
	
	Since $T(\tau^*,\mu^*)\ne \tau^*-1$ this means that one of the bounds \eqref{eq:T_UB} or \eqref{eq:T_LB} would be violated by the choice $T=\tau^*-1$. Since the bound is continuous in $\tau$, there is some open interval $I$ around $\tau^*$ such that for all $\tau\in I$, the same bound \eqref{eq:T_UB} or \eqref{eq:T_LB} would still be violated by the choice $T=\tau-1$. So for all $\tau\in I$, the according bound is tight, i.e.,
	\begin{align*}
	T(\tau,\mu^*)= \ln\frac{\rho-i\mu^*}{\mu^* \tau+\rho-\mu^*-1}\qquad \text{ for all $\tau\in I$},
	\end{align*}
	with $i=1$ or $i=2$ depending on which of the two bounds would have been violated.
	Using \eqref{eq:costtaumu} we get
	\begin{align*}
	\cost(\tau,\mu^*)&=\rho\tau+(T(\tau,\mu^*)-\tau)(\rho-i\mu^*)\\
	&=i\mu^* \tau+T(\tau,\mu^*)\cdot(\rho-i\mu^*)\qquad\text{ for all $\tau\in I$}.
	\end{align*}
	
	We next compute the derivatives with respect to $\tau$:
	\begin{align*}
	\frac{d}{d \tau}T(\tau,\mu^*)&=\frac{-\mu^*}{\mu^* \tau+\rho-\mu^*-1}\qquad \text{ for all $\tau\in I$}
	\intertext{and therefore}
	\frac{d}{d\tau}\cost(\tau,\mu^*)&=i\mu^*-\frac{\mu^*}{\mu^* \tau+\rho-\mu^*-1}\cdot(\rho-i\mu^*)\\
	&= \frac{i\mu^*}{\mu^* \tau+\rho-\mu^*-1}\left(\mu^* \tau+\frac{i-1}{i}\rho-1\right)\qquad \text{ for all $\tau\in I$}.
	\end{align*}
	
	Thus, inside $I$, we see that $\cost(\tau,\mu^*)$ is decreasing in $\tau$ for $\tau<\left(1-\frac{i-1}{i}\rho\right)/\mu^*$ and increasing in $\tau$ for $\tau>\left(1-\frac{i-1}{i}\rho\right)/\mu^*$. In particular, $\cost(\tau,\mu^*)$ is maximized when $\tau$ approaches one of the two boundaries of $I$. So we can find $\tau\in I$ with $\cost(\tau,\mu^*)>\cost(\tau^*,\mu^*)$, as desired.
\end{proof}

To obtain an upper bound on $\mu_\tau$, the previous lemma now allows us to eliminate $\tau$ from our equations by either replacing it by $T-1$ or letting it tend towards $1$.

\begin{proof}[Proof of Lemma~\ref{lem:muBoundCase3}]
	For $\mu=\mu_\tau$ and $T=T(\tau,\mu_\tau)$, we get from Lemma~\ref{lem:rhoTight} and equation \eqref{eq:costtaumu} that
	\begin{align*}
	\rho=\rho\tau- (\tau-T)(\mu\tau+\rho-\mu-1)e^{T}.
	\end{align*}
	Rearranging to isolate $\mu$, we get
	\begin{align}
	\mu&=\frac{\rho(\tau-1)-(\tau-T)(\rho-1)e^T}{(\tau-T)(\tau-1)e^T}\nonumber\\
	&=\frac{\rho}{(\tau-T)e^T}-\frac{\rho-1}{\tau-1}.\label{eq:muFge1}
	\end{align}
	Since $\mu\ge 0$, it must be the case that $T(\tau,\mu_\tau)\to 1$ as $\tau\to 1$, so the lower bound \eqref{eq:T_LB} is tight for $\tau$ close to $1$. Writing \eqref{eq:T_LB} with equality and plugging in the limit values $\tau=T=1$, this shows that
	\begin{align*}
	\frac{\rho-2\mu_\tau}{\rho-1}\to e\qquad\text{as }\tau\to1,
	\end{align*}
	so 
	\begin{align*}
	\mu_\tau\to \frac{\rho+e-\rho e}{2}\qquad\text{as }\tau\to1.
	\end{align*}
	If $\mu_\tau$ is \emph{not} maximized when $\tau\to 1$, then let $\tau>1$ be such that $\mu_\tau=\max_{\tau>1}\mu_\tau$. By Lemma~\ref{lem:localMax}, we then have $\tau=T+1$. Plugging this into \eqref{eq:muFge1}, we get
	\begin{align*}
	\mu_\tau&=\frac{\rho}{e^T}-\frac{\rho-1}{T}.
	\end{align*}
	Taking derivatives shows that the right hand side is maximized when $T\in(0,1]$ is the solution to $T^2e^{-T}=\frac{\rho-1}{\rho}$. Then the $\frac{\rho-1}{T}$ term is equal to $\rho Te^{-T}$, so that overall $\mu\le \rho(1-T)e^{-T}$.
\end{proof}

\section{Proof of Theorem~\ref{thm:ski-LB}: Lower bound for ski rental}
\label{sec:proofLB}

We will show for any fixed $\tau$ that one cannot be better than $(\rho,\mu_\tau(\rho))$-competitive (with $\mu_\tau$ as defined in Section~\ref{sec:defMuTau}). This will imply Theorem~\ref{thm:ski-LB} due to the following lemma:

\begin{lemma}
	\label{lem:worst-mu}
	For any $\rho\in\left[1,\frac{e}{e-1}\right]$ there exists $\tau\ge 0$ such that $\mu_\tau(\rho) = \mu(\rho)$.
\end{lemma}
\begin{proof}
	Note that the two expressions in the definition \eqref{eq:mu} of $\mu(\rho)$
	yield the same value for $\tilde\rho \approx 1.16$.
	
	Whenever $\rho \geq \tilde\rho$, the first term in the definition \eqref{eq:mu} dominates. In this case, we choose $\tau = \ln 2$.
	Such $\tau$ belongs to the second case of the definition \eqref{eq:mutau}
	of $\mu_\tau(\rho)$ whenever $(1-\tau)e^{\tau-1} \leq 2-\frac2\rho$.
	This is equivalent to $\rho \geq \frac{e}{e- 1 + \ln 2} \approx 1.127$,
	which is in fact smaller than $\tilde \rho$.
	Therefore, for any $\rho \geq \tilde\rho > 1.127$, we have
	\[ \mu_{\tau}(\rho)
	= \frac{e^{\tau} (1-\rho\frac{e-1}{e})}{2 - (1-\tau)e^\tau}
	= \frac{1-\rho\frac{e-1}{e}}{\ln 2} = \mu(\rho). \]
	
	Now, consider $\rho < \tilde\rho$. Here, the second term in the definition \eqref{eq:mu} of $\mu(\rho)$ dominates. For such $\rho$, there is a (unique) $T \in \left[0,\frac{1}{2}\right)$ such that
	$T^2 e^{-T} = 1 - \frac1\rho$. We choose $\tau = 1-T$. This belongs to the first case of the definition \eqref{eq:mutau} of $\mu_\tau(\rho)$, since $(1-\tau)e^{\tau-1} > 2-\frac2\rho$ is equivalent to $1 - \frac1\rho > T(2-\frac2\rho)$, which is satisfied since $T<\frac{1}{2}$.
	Thus,
	\begin{align*}
	 \mu_\tau(\rho) &= -\frac{\rho-1}{1- \tau}+\rho e^{ \tau-1}\\
	&= \rho\left(- \frac{T^2 e^{-T}}{T} + e^{-T}\right) \\
	&= \rho(1-T) e^{-T} \\
	&= \mu(\rho),
	\end{align*}
	where the second equation uses the definition of $T$ and the choice of $\tau$.
\end{proof}

The goal in the remainder of this section is to prove the following theorem, which combined with Lemma~\ref{lem:worst-mu} yields Theorem \ref{thm:ski-LB}.
\begin{theorem}\label{thm:ski-LB-tau}
	For all $\rho \in \left[1, \frac{e}{e-1}\right]$, $\tau\ge 0$ and $\widehat\mu<\mu_\tau(\rho)$, there exists no $(\rho,\widehat\mu)$-competitive algorithm for ski rental instances with prediction $\tau$.
\end{theorem}

Let $\rho\in\left(1,\frac{e}{e-1}\right]$ and a prediction $\tau\ge 0$ be fixed for the remainder of this section and write $\mu_\tau=\mu_\tau(\rho)$. Note that by continuity, it suffices to show the theorem for $\rho>1$.

The idea to prove Theorem~\ref{thm:ski-LB-tau} is to show that any algorithm must have at least the same expected cost as our algorithm on some instance for which the analysis of our algorithm is tight. Denote by $\cost(x)$ the expected cost of our algorithm for $\mu=\mu_\tau$ when the skiing season ends at time $x$, and denote by $\opt(x)=\min\{x,1\}$ the optimal cost and by $\eta(x)=|\tau-x|$ the prediction error. The following lemma summarizes the cases for which the analysis of our algorithm is tight.

\begin{lemma}\label{lem:algoTight}
	For $\mu=\mu_\tau$, our algorithm satisfies $\cost(x)=\rho\opt(x)+\mu_t\eta(x)$ in the following cases:
	\begin{itemize}
		\item In Case 1 for all $x\ge \tau$.
		\item In Case 2 for all $x\ge 0$.
		\item In Case 3 for all $x\in [0,T]\cup\{\tau\}$, and if $P_\infty=\mu$, then also for all $x > \tau$.
	\end{itemize}
\end{lemma}
\begin{proof}
\begin{description}
	\item[Case 1:] By Observation~\ref{obs:tight1}, it suffices to show that $b=\tau$ and $P_\infty=\mu$. 
	By definition of $\mu=\mu_\tau$, we have
	\begin{align*}
	P_0+\mu+\int_\tau^1\rho e^{t-1} dt&=\frac{\tau(\rho-1)}{1-\tau}-\frac{\rho-1}{1- \tau}+\rho e^{ \tau-1} + \rho-\rho e^{\tau-1}\\
	&=1.
	\end{align*}
	This implies $P_\infty=\mu$ since otherwise we would have $1=P_0+P_\infty<P_0+\mu\le1$, a contradiction. By definition of $b$, this also yields $b=\tau$.
	\item[Case 2:] By Observation~\ref{obs:tight2}, it suffices to show that $a=b=\tau$ and $P_\infty=\mu$. We have
	\begin{align*}
	P_0+\mu+ \int_{0}^\tau (\mu\tau+\rho-\mu-1)e^t dt
	+ \int_{\tau}^1 \rho e^{t-1} dt&= \mu\tau+\mu+ (\mu\tau+\rho-\mu-1)(e^\tau-1)
	+ \rho -\rho e^{\tau-1} dt\\
	&= \mu\left(2+(\tau-1)e^\tau)\right) +1+e^\tau(\rho-1-\rho/e)\\
	&= e^{ \tau}\left(1-\rho\frac{e-1}{e}\right) + 1+e^\tau(\rho-1-\rho/e)\\
	&= 1,
	\end{align*}
	implying $P_\infty=\mu$ by the same argument as before, and therefore also $a=b=\tau$ by definition of $a$ and $b$.
	\item[Case 3:] Follows from Observation~\ref{obs:tight3} and Lemma~\ref{lem:rhoTight}.\qedhere
\end{description}
\end{proof}

Let $\alg$ denote our algorithm (for prediction $\tau$ and $\mu=\mu_\tau$) and $F$ its CDF for the time of buying. Note that $F$ is continuous.

Suppose that a $(\rho,\widehat\mu)$-competitive algorithm $\widehat\alg$ for instances with prediction $\tau$ exists for some $\widehat\mu<\mu_\tau(\rho)$, and denote by $\widehat F$ its CDF. We may assume without loss of generality that $\widehat F$ is continuous (otherwise, $\widehat F$ can be approximated arbitrarily well by a continuous CDF whose corresponding algorithm is $(\rho,\widehat\mu')$-competitive for $\widehat\mu'$ arbitrarily close to $\widehat\mu$). Denote by $ \widehat\cost(x)$ the expected cost of $\widehat\alg$ when the skiing season has length $x$. To prove Theorem~\ref{thm:ski-LB-tau}, it suffices to show that there exists $x$ for which $\cost(x)$ is tight (as per Lemma~\ref{lem:algoTight}) and $\widehat\cost(x)>\cost(x)$.

Let $y=\inf\{t\ge 0\colon \widehat F(t)\ne F(t)\}$. Note that $y<\infty$ must exist, since otherwise $ \widehat\alg$ is the same as $\alg$,
which cannot achieve $\widehat\mu < \mu_\tau(\rho)$ by Lemma~\ref{lem:algoTight}. By continuity of $ \widehat F$ and $F$, we either have $y=0$
or $ \widehat F(t) = F(t)$ for all $t\leq y$. Also by continuity, there exists $y'>y$ such that either $\widehat F(t) > F(t)$ for all $t\in(y,y']$ or  $\widehat F(t) < F(t)$ for all $t\in(y,y']$. We may assume $y'<y+1$.

\begin{lemma}
\label{lem:LB-greater}
If $\widehat F(t) > F(t)$ for all $t\in(y,y']$, then there exists $z>y$ such that $\widehat\cost(z)>\cost(z)$.
\end{lemma}
\begin{proof}
If $y=0$ and $ \widehat F(0) > F(0)$, then $\widehat\cost(0) > \cost(0)$ and the desired $z>0$ must exist by continuity of $\cost$ and $\widehat\cost$.

Otherwise, $\widehat F(t) = F(t)$ for all $t\leq y$  and therefore the expected cost incurred until time $y$ is the same for both algorithms. Let $z=\arg\max_{z\in(y,y']} \{\widehat F(z)-F(z)\}$. In the time interval $(y,z]$, the expected buying cost incurred by $\widehat\alg$ is precisely $\widehat F(z)-F(z)$ greater than that of $\alg$, and the rental cost saved by $\widehat\alg$ compared to \alg is at most $(z-y)(\widehat F(z)-F(z))$. Since $z-y\le y'-y<1$, this saving in rental cost cannot make up for the greater buying cost, and therefore $\widehat\cost(z)>\cost(z)$.
\end{proof}

\begin{lemma}
\label{lem:LB-smaller}
If $\widehat F(t) < F(t)$ for all $t\in(y,y']$,
then there is $z>y$ such that $ \widehat\cost(z) > \cost(z)$.
In particular, we can choose $z=\min\{z>y\mid \widehat F(z) = F(z)\}$ if it exists.
Otherwise, $z$ can be chosen very large so that, e.g., $z > \tau$.
\end{lemma}
\begin{proof}
If $ \widehat F(t) < F(t)$ for all $t > y$, then for $z$ large enough it holds
that $ \widehat\cost(z) > \cost(z)$, since $\widehat \alg$'s greater rental cost will eventually outweigh its smaller buying cost.
Clearly, if we choose $z$ sufficiently large, we also get the additional property
that $z > \tau$.

Otherwise, by continuity we can choose $z=\min\{z>y\mid \widehat F(z) = F(z)\}$.
Since $\widehat F(z) = F(z)$, the expected buying cost is the same for both algorithms.
On the other hand, since $\widehat F(t) < F(t)$ for all $t \in (y,z)$,
i.e., $ \widehat\alg$ is delaying buying, the expected rental cost of $ \widehat\alg$
is higher than that
of $\alg$.
\end{proof}

To complete the proof of Theorem~\ref{thm:ski-LB-tau}, we will show that in each of the three cases from the definition of our algorithm, there exists an $x$ satisfying the condition of Lemma~\ref{lem:algoTight} and such that $\widehat\cost(x)>\cost(x)$ or, when $x\ne\tau$ (so that $\eta>0$), $\widehat\cost(x)\ge\cost(x)$.

\begin{description}
	\item[Case 1:] By Lemma~\ref{lem:algoTight}, $\cost(x)$ is tight for all $x\geq \tau$.
	
	Recall from the definition of our algorithm that $F(0)=F(\tau)$. 
	
	If $\widehat F(\tau)>F(0)=F(\tau)$, then $\widehat\alg$'s expected buying cost up until time $\tau$ exceeds that of $\alg$ by exactly $\widehat F(\tau)-F(0)$, and its reduction in renting cost is at most $\tau(\widehat F(\tau)-F(0))$. Since $\tau<1$ in Case 1, the reduction cannot make up for the excess and we get $\widehat\cost(\tau)>\cost(\tau)$.
	
	Otherwise, $\widehat F(\tau)\le F(\tau)$.
	Lemmas~\ref{lem:LB-greater} and~\ref{lem:LB-smaller}
	give us $z$ such that $ \widehat\cost(z) > \cost(z)$.
	If $z \geq \tau$,
	we are done. Otherwise, we have
	\[  \widehat\cost(\tau) \geq  \widehat\cost(z) + (\tau-z)(1- \widehat F(\tau))
	> \cost(z) + (\tau-z)(1-F(\tau)) = \cost(\tau),
	\]
	since $ \widehat\cost(z) > \cost(z)$ and $ \widehat F(\tau) \leq F(\tau)$.
	\item[Case 2:] By Lemma~\ref{lem:algoTight}, $\cost(x)$ is tight for all $x\geq 0$. Thanks to Lemmas~\ref{lem:LB-greater} and~\ref{lem:LB-smaller}, we are done.
	
	\item [Case 3:] By Lemma~\ref{lem:algoTight}, $\cost(x)$ is tight for all $x\in [0,T]\cup\{\tau\}$, and if $P_\infty=\mu$, then also for all $x > \tau$.
	
	If $\widehat F(0)\ge F(0)$, then we choose $x=0\ne\tau$, noting that $\widehat\cost(0)=\widehat F(0)\ge F(0)=\cost(0)$.
	
	Otherwise $\widehat F(0)< F(0)$. 
	Lemma~\ref{lem:LB-smaller} gives us $z$ such that
	$ \widehat\cost(z) > \cost(z)$.
	Recall that $T$ is chosen within the interval
	defined by \eqref{eq:T_UB} and \eqref{eq:T_LB} as close as possible to
	$\tau - 1$. Thus, we observe that $T < \tau-1$ only if \eqref{eq:T_UB} is tight and
	$P_\infty=0$. Similarly, $T>\tau-1$ only if \eqref{eq:T_LB}
	is tight and $P_\infty = \mu$. We consider two cases.
	\begin{itemize}
		\item $\widehat F(\tau) \geq F(\tau)$. By continuity of
		$\widehat F$ and Lemma~\ref{lem:LB-smaller}, we have
		$z\leq \tau$ and $\widehat F(z) = F(z)$. If
		$z\in[0,T]\cup\{\tau\}$, we are done, so we may assume that
		$T<z<\tau$. By definition of $T$,
		either $P_\infty=0$ or $\tau-T \leq 1$.
		If $P_\infty = 0$, then $ \widehat F(t) = F(t) = 1$ for any $t\geq z$
		and we have $ \widehat\cost(\tau) =  \widehat\cost(z)
		> \cost(z) = \cost(\tau)$.
		Otherwise, by considering the additional buying and
		rental cost after time $z$, we have
		\begin{align*}
		\widehat\cost(\tau) &\geq \widehat\cost(z) + \big( \widehat F(\tau)
		- \widehat F(z)\big) + (\tau - z)\big(1- \widehat F(\tau)\big)\\
		&> \cost(z) + (\tau-z)\big( \widehat F(\tau) - \widehat F(z)\big)
		+ (\tau-z)\big(1-\widehat F(\tau)\big)\\
		&= \cost(z) + (\tau-z)(1-F(\tau))\\
		&= \cost(\tau),
		\end{align*}
		where the strict inequality follows from $\widehat\cost(z) > \cost(z)$ and $\tau-z < \tau-T \leq 1$,
		and the equations use $\widehat F(z) = F(z) = F(\tau)$, recalling that our algorithm does not buy after time $T$.
		
		\item $\widehat F(\tau) < F(\tau)$. If there exists $t\in[0,T]$ such that $\widehat F(t)=F(t)$, then $z\in[0,T]$ by Lemma~\ref{lem:LB-smaller} and we are done. Otherwise, using $\widehat F(\tau) < F(\tau)$, monotonicity and continuity of the CDFs, and the fact that $F(t)$ is constant for $t\ge T$, we get that $\widehat F(t)<F(t)$ for all $t\in[0,\tau]$. Then $z>\tau$ by Lemma~\ref{lem:LB-smaller}. If $P_\infty = \mu$, then
		$\cost(z)$ is tight and
		therefore it is enough to consider the case $P_\infty < \mu$, where
		$\tau-T \geq 1$ by definition of $T$.
		
		In this case, we choose $x=\tau$. To compare $\widehat\cost(\tau)$ and $\cost(\tau)$, we will separately consider the contribution coming from buying cost, rental cost in the interval $[0,T)$ and rental cost in the interval $[T,\tau]$. Since $\widehat F(t)\le \widehat F(\tau)<F(\tau)=F(t)$ for all $t\in[T,\tau]$, the rental cost of $\widehat\alg$ in the time interval $[T,\tau]$ exceeds that of $\alg$ by at least $(\tau-T)\left(F(\tau)-\widehat F(\tau)\right)\ge F(\tau)-\widehat F(\tau)$, where the inequality is due to $\tau-T \geq 1$. The benefit in buying cost of $\widehat\alg$ compared to $\alg$, which amounts to $F(\tau)-\widehat F(\tau)$, can at best make up for this. Since $\widehat F(t)<F(t)$ for all $t\in[0,\tau]$, the rental cost of $\widehat\alg$ in the remaining time interval $[0,T)$ is strictly greater than that of $\alg$, so that overall $\widehat\cost(\tau)>\cost(\tau)$. Here, we used that $T>0$ by Fact~\ref{fact:T0}.\footnote{This implicitly assumes that $\rho>1$, since in Section~\ref{sec:Case3Comp} we assumed $\mu<1/\tau$, which in Section~\ref{sec:boundingMu3} was proved to be true for $\mu=\mu_\tau$ assuming $\rho>1$. However, by continuity Theorem~\ref{thm:ski-LB-tau} also extends to the case $\rho=1$. Alternatively, we can get strict inequality even for $T=0$ by noting that then $\tau-T=\tau>1$.}
	\end{itemize}
\end{description}

\section{Reduction from DPM to ski rental}
\label{sec:DPM-1}

We now give a reduction from DPM to ski rental in the learning-augmented setting (Lemma~\ref{lem:redDPMSki}), provided that the ski rental algorithm satisfies the following monotonicity property: We say that a ski rental algorithm for rental cost $\alpha=1$ and buying cost $\beta=1$ is \emph{monotone} if its CDF $F_{\tau}$ for the buying time when given prediction $\tau$ satisfies
\begin{align*}
F_{\tau}(t)\le F_{\tau'}(t) \qquad \text{for all }t \geq 0\text{ and }\tau<\tau'.
\end{align*}

Intuitively, this property is very natural: The longer the predicted duration of skiing, the greater should be our probability of buying. Indeed, our algorithm satisfies this property:
\begin{lemma}\label{lem:monotone-cdf}
For $\mu=\mu(\rho)$, the $(\mu,\rho)$-competitive ski rental algorithm
from Section~\ref{sec:ski} is monotone, i.e., its CDF $F_\tau$ when given
prediction $\tau$ satisfies $F_{\tau}(t)\le F_{\tau'}(t)$
for all $t \geq 0$ and $\tau < \tau'$.
\end{lemma}
\begin{proof}
One can verify that at the boundary between adjacent cases of the algorithm, both cases define the same probability distribution (i.e., $\tau =1-\frac{\rho-1}{\mu}$ as the boundary between Case 1 and Case 2, and $\tau=1$ as the boundary between Case 2 and Case 3). In particular, for $\tau=1$ this can be seen by separately considering $\mu\ge 1/2$ (where Case 3 would give $T=0$) and $\mu<1/2$ (which makes~\eqref{eq:T_LB} tight, so that Case 3 would give $P_\infty=\mu$). It remains to show that $F_{\tau}(t)$ is non-decreasing as $\tau$ increases \emph{within} one of the three cases.

\textbf{Cases 1 and 2} ($\boldsymbol{\tau \leq 1}$)\textbf{.}
Here, $P_0$ is increasing in $\tau$ and so is
$(\tau\mu + \rho - \mu -1)e^t$.
Therefore, $F_{\tau}(t) = P_0 + \int_0^t p_t dt$
is increasing in $\tau$ as long as $t\leq b$.
If $F_{\tau}(\tau) \geq 1-\mu$, then $p_t=0$ for each $t > \tau$ and
we are done.
Otherwise, note that $F_{\tau}(t) = 1-\mu$ for any $t\geq1$. And, going
backwards in time, $F_{\tau}(t)$ has the same value for any $\tau$ until
it reaches $F_{\tau}(b)$ at $t=b$.

\textbf{Case 3} ($\boldsymbol{\tau> 1}$)\textbf{.}
Both $P_0$ and $(\tau\mu + \rho - \mu -1)e^t$ are increasing in $\tau$,
therefore $F_{\tau}(t)$ is increasing in $\tau$ until
it reaches $1-P_\infty$ at $t=T$.
Therefore, it is enough to show that $P_\infty$ is
non-increasing in $\tau$.
Either one of \eqref{eq:T_UB} and \eqref{eq:T_LB} is tight,
or we have $T = \tau-1$.
In the first case, $P_\infty$ is equal to $0$ or $\mu$, i.e., independent
on $\tau$. Otherwise, $T=\tau-1$ implies that $e^T$
is increasing in $\tau$, and therefore
$P_\infty$ is decreasing in $\tau$.
\end{proof}

As mentioned earlier, for many $\tau$ one could actually achieve a better $\mu_\tau(\rho)<\mu(\rho)$. However, somewhat surprisingly the optimal such algorithm would \emph{not} be monotone.
The monotonicity of our algorithm therefore crucially relies on our specific description (in particular the choice of $a$ and $b$), which only aims for $(\rho,\mu(\rho))$-competitiveness with $\mu(\rho)=\sup_\tau\mu_\tau(\rho)$.

Combining Theorem~\ref{thm:ski-UB}, Lemma~\ref{lem:redDPMSki} and Lemma~\ref{lem:monotone-cdf}, we get:
\begin{coro}\label{cor:rhoMuDPM}
	For every $\rho\in[1,\frac{e}{e-1}]$, there is a $(\rho,\mu(\rho))$-competitive algorithm for DPM.
\end{coro}

To prove Lemma~\ref{lem:redDPMSki}, it suffices to describe a $(\rho,\mu)$-competitive algorithm for the special case of DPM with a single idle period: Running such an algorithm for each individual period yields a $(\rho,\mu)$-competitive algorithm for DPM with any number of idle periods, since we can simply sum inequality \eqref{eq:rhomu} over all periods to obtain the corresponding inequality for the entire instance.

Consider now a single idle period of length $\ell$ for DPM. We first recall some observations of \citet{IraniSG03} about the optimal \emph{offline} algorithm: It is easy to see that the optimal offline algorithm would transition to some state $j$ only once at the beginning of the period and remain there throughout the period, paying cost $\alpha_j\ell+\beta_j$. Thus, state $j$ is preferred over state $j-1$ if and only if $\alpha_{j-1} \ell + \beta_{j-1} > \alpha_j\ell + \beta_j$, or equivalently
$\ell > t_j := \frac{\beta_j-\beta_{j-1}}{\alpha_{j-1}-\alpha_j}$. We may assume without loss of generality that $t_1 < \dots <  t_k$: Indeed, suppose
$t_{j+1} \le  t_j$, then state $j$ is redundant because whenever $j$ is preferred over $j-1$, then $j+1$ is preferred over $j$.
Defining $t_0:=0$ and $t_{k+1}:=+\infty$, we get
a partition $[0,+\infty) = \bigcup_{j=0}^k I_j$, where
$I_j = [t_j, t_{j+1})$.
We can then express the cost of the offline optimum as
\begin{equation}
\label{eq:dpm-1-opt}
\opt = \alpha_{j^*} \ell + \beta_{j^*},
	\text{ with $j^*$ such that } \ell \in I_{j^*}.
\end{equation}

In the online setting, we of course do not know $\ell$.
The idea of our algorithm (similar to~\cite{LotkerPR12}) is to simulate $k$ ski rental algorithms $\alg_1,\dots,\alg_k$
in parallel, where the task of $\alg_j$ is to decide whether it is time
to transition from the state $j-1$ to $j$. For this, we choose $\alg_j$ to be an algorithm for ski rental with rental cost $\alpha_{j-1}-\alpha_j$ and buying cost $\beta_j-\beta_{j-1}$. Let $F_{\tau}$ be the CDF of the buying time of a \emph{monotone} ski rental algorithm (for $\alpha=\beta=1$) when given prediction $\tau$. Recalling our reduction from arbitrary $\alpha$ and $\beta$ to the case $\alpha=\beta=1$ in Lemma~\ref{lem:unit-costs}, the CDF of $\alg_j$ is given by
\begin{align}
F^{j}(t):=F_{\tau/t_j}\left(t/t_j\right).\label{eq:Fj}
\end{align}
An outline of our algorithm is given in Algorithm~\ref{alg:DPM-1}.

\begin{algorithm2e}
\caption{DPM with a single idle period}
\label{alg:DPM-1}
\For{j=1,\dots,k}{
	Let $F^j$ be as defined by~\eqref{eq:Fj}, induced by a monotone $(\rho,\mu)$-competitive ski rental algorithm\;
}
Choose $p \in [0,1]$ uniformly at random\;

At any time $t$: choose state $j = \max\{j\colon F^j(t)\geq p\}$\;
\end{algorithm2e}

\paragraph{Proof of Lemma~\ref{lem:redDPMSki}.}

We now show that Algorithm~\ref{alg:DPM-1} is $(\rho,\mu)$-competitive for DPM instances with a single idle period, completing the proof of Lemma~\ref*{lem:redDPMSki}.

Denote by $A_j:=\inf\{t\,\vert\; F^j(t) \geq p\}$ the random variable for the time when $\alg_j$ triggers the transition to state $j$. As shown in Lemma~\ref{lem:monotone-cdf}, the monotonicity of the ski rental algorithm yields $F^{j-1}(t)\ge F^j(t)$ for all $t$. Therefore, $A_1 \leq \dotsb \leq A_k$.

For a prediction $\tau$ and an idle period of length $x=\ell$, let
\[ \cost_j=(\beta_j-\beta_{j-1})
	\cdot \1_{A_j\le x}+(\alpha_{j-1}-\alpha_j)\cdot\min\{A_j,x\}
\]
be the random variable for the cost of $\alg_j$.
Since $\alg_j$ is $(\rho,\mu)$-competitive,
its expected cost is at most
\begin{align*}
E(\cost_j) \le
\rho \cdot \min\{ (\alpha_{j-1}-\alpha_j) x, \beta_j-\beta_{j-1} \}
	+ \mu \cdot (\alpha_{j-1}-\alpha_j) \cdot |\tau-x|.
\end{align*}

Let $\bj$ be the maximal $j$ such that $A_j\le x$. Writing $A_0:=0$ and
$A_{k+1}:=\infty$, the cost of Algorithm~\ref{alg:DPM-1} is
\begin{align*}
\cost &= \beta_{\bj} + \sum_{j=1}^{\bj} (A_j-A_{j-1})\alpha_{j-1} + (x-A_{\bj})\alpha_{\bj}\\
&= \sum_{j=1}^k (\beta_j-\beta_{j-1})\cdot \1_{A_j\le x}+\sum_{j=1}^{k+1}(\min\{A_j,x\}-\min\{A_{j-1},x\})\alpha_{j-1}\\
&= \sum_{j=1}^k (\beta_j-\beta_{j-1})\cdot \1_{A_j\le
x}+\sum_{j=1}^k(\alpha_{j-1}-\alpha_j)\min\{A_j,x\}+\alpha_kx,
\end{align*}
which is precisely $\alpha_k x + \sum_{j=1}^k\cost_j$.
Taking expectations, we get
\begin{align*}
E(\cost)
&\le \alpha_k x + \rho\sum_{j=1}^k \min \big\{ (\alpha_{j-1}-\alpha_j) x,
	\beta_j-\beta_{j-1}\big\}
	+ \mu\cdot|\tau-x|\sum_{j=1}^k(\alpha_{j-1}-\alpha_j)\\
&= \alpha_k x + \rho \big((\alpha_{j^*}-\alpha_k)x + \beta_{j^*}-\beta_{0} \big)
	+ \mu \cdot |\tau-x| \cdot (\alpha_0-\alpha_k),
\end{align*}
where $j^*$ denotes the optimal state such that
$\opt = \alpha_{j^*} x + \beta_{j^*}$.
The equality holds since
\begin{align*}
\beta_j-\beta_{j-1}\le x\cdot(\alpha_{j-1}-\alpha_j)\iff t_j\le x \iff j \leq j^*.
\end{align*}
The theorem now follows since $\beta_0=0$, $\alpha_k \geq 0$ and $1\le \rho$.
\qed

\subsection{Conversion to a prudent algorithm}
\label{section:prudent}

It was shown by~\citet{LotkerPR12} that any DPM algorithm can be converted (online) into one that assigns a non-zero probability to at most two adjacent power states, and the resulting algorithm can only have smaller expected cost than the original algorithm. More precisely, for a given time, denote by $p_i$ the probability that the algorithm is in power state $i$. We call the probability vector $(p_0,\dots,p_k)$ \emph{prudent} if it is of the form $(0,\dots,0,p_{i-1},p_i,0,\dots,0)$ for some $i=1,\dots k$. A DPM algorithm is prudent if its probability vector is prudent at all times. Note that our DPM algorithm is \emph{not} prudent: Indeed, our ski rental algorithm typically buys with non-zero probability already at time $0$, and thus the DPM algorithm obtained via the above reduction assigns non-zero probability to \emph{all} power states at time $0$. Since conversion to a prudent algorithm using the method of~\cite{LotkerPR12} reduces an algorithm's cost, any implementation should apply this conversion, which we now describe.

To make the DPM algorithm prudent, we need to convert each non-prudent probability vector $p=(p_0,\dots,p_k)$ into a prudent one $\tilde p$. Let $B_p=\sum_i p_i\beta_i$ be the expected wake-up cost of vector $p$, and let $m=\max\{i\mid \beta_i\le B_p\}$ be the deepest state whose wake-up cost is at most $B_p$. Note that if $m=k$, then $p=(0,\dots,0,1)$ is already prudent. Otherwise, $m<k$ and define
\begin{align*}
\tilde p_i = \begin{cases}
\frac{\beta_{m+1}-B_p}{\beta_{m+1}-\beta_m}\quad&i=m\\
1-\tilde p_m&i=m+1\\
0&\textit{otherwise.}
\end{cases}
\end{align*}
It was shown in \cite[Theorem 4.2]{LotkerPR12} that when replacing any non-prudent vector $p$ by a prudent vector $\tilde p$ in this way, the resulting algorithm pays the same expected wake-up cost but less running cost than the original non-prudent algorithm.

\section{Finding the best trade-off online}
\label{section:BlumBurch}

Our goal is to design an algorithm whose performance almost matches that of Corollary~\ref{cor:rhoMuDPM} simultaneously for \emph{all} $\rho$,
proving Theorem~\ref{thm:DPM}.
It will be useful to view DPM as a Metrical Task System.

\paragraph{Metrical Task Systems (MTS).}
Metrical Task Systems (MTS), introduced by \citet{BorodinLS92},
is a broad class of online problems containing many other problems as special cases.
In MTS, we are given a metric space $M$ of {\em states}.
We start at a predefined initial state $x_0$.
At each time $t = 1, 2, \dotsc, T$, we are presented with
a {\em cost function} $c_t\colon M \to \R_+$.
Then, we have to choose our new state $x_t$ and pay
$\dist(x_{t-1}, x_t) + c_t(x_t)$, where $\dist(x_{t-1}, x_t)$ is the distance
between $x_{t-1}$ and $x_t$ in $M$.
The objective is to minimize the overall cost incurred over time.

To formulate DPM as a Metrical Task System, we choose states
$0, 1, \dotsc, k$ corresponding to the power states, with distances
$\dist(i, j) = \frac12 |\beta_i - \beta_j|$, so that
the cost of switching from the state 0 to $j$ and back is
$\beta_j$. We choose $0$ as the initial state.
We discretize time
in the DPM instance using time steps of some small length $\delta>0$. At each time step
belonging to some idle period, we issue a cost function
$c$ such that $c(j) = \delta \alpha_j$ for each $j = 0, \dotsc, k$.
At the end of each idle period, we issue a cost function
where $c(0) = 0$ and $c(j) = +\infty$ for $j=1, \dotsc, k$,
which forces any algorithm to move back to the active state.

We use the result of \citet{BlumB00}
to combine multiple instances of our algorithm with
different parameters $\rho$. 

\begin{theorem}[\citet{BlumB00}]
\label{thm:BB}
There is an algorithm which,
given $N$ online algorithms $A_1,\dots A_N$ for an MTS with diameter $D$
and $\epsilon_1 < 1/2$, achieves expected cost at most
\[ (1+\epsilon_1)\cdot \min_i \{\cost(A_i)\} + O(D/\epsilon_1)\ln N. \]
\end{theorem}

\paragraph{Proof of Theorem~\ref{thm:DPM}}
For $\epsilon_2 >0$, we choose a set $P \subset [1,\frac{e}{e-1}]$
of $O(1/\epsilon_2)$ values of $\rho$ as follows:
\[ P := \left\{1,\tilde\rho,\frac{e}{e-1}\right\}
	\cup \left\{ \rho_i\,\Big\vert\;
	\mu(\rho_i) = (1+i\epsilon_2) \mu(\tilde\rho), \text{ where }
	i = 1, \dotsc,
	\left\lfloor\frac{1-\mu(\tilde\rho)}{\mu(\tilde\rho)\epsilon_2}\right\rfloor
	\right\},
\]
where $\tilde\rho\approx 1.16$ is as defined below equation \eqref{eq:mu}.
We claim that
\[ \min_{\rho \in P} \{\rho\opt + \mu(\rho)\eta\}
\leq (1+\epsilon_2) \min_{\rho \in [1, e/(e-1)]} \{\rho\opt + \mu(\rho)\eta\}.
\]
If the minimizer $\rho$ of the right-hand side lies in
$\left[\tilde\rho, \frac{e}{e-1}\right]$, then it is one of the endpoints $\rho\in\left\{\tilde\rho,\frac{e}{e-1}\right\}$ since $\mu(\rho)$ is linear in this interval.
If the minimizer is $\rho \in [1, \tilde\rho)$,
then there exists $i$ such that $\rho_- \leq \rho \leq \rho_i$, where $\rho_-=\rho_{i+1}$ if it exists and $\rho_-=1$ otherwise. Here, we have
$\mu(\rho_-) \leq (1+\epsilon_2)\mu(\rho_i) \leq (1+\epsilon_2)\mu(\rho)$,
since $\mu$ is a decreasing function. Thus, $\rho_-\opt+\mu(\rho_-)\eta\le \rho\opt + (1+\epsilon_2)\mu(\rho)\eta$.

Combining our $(\rho,\mu(\rho))$-competitive algorithms for DPM for $\rho \in P$ using Theorem~\ref{thm:BB}, the resulting algorithm has cost at most
\[ (1+\epsilon_1)(1+\epsilon_2)
	\min_{\rho \in [1, e/(e-1)]} \big\{\rho\opt + \mu(\rho)\eta\big\}
	+ O\bigg(\frac{\beta_k}{\epsilon_1} \cdot \ln \frac1{\epsilon_2}\bigg),
\]
since $D=\beta_k/2$ in our case.
This concludes the proof.\qed

\subsection{Remarks on shifting/dynamic regret}

Instead of the connection to MTS and the result of \citet{BlumB00},
we could also use other online learning techniques.
In particular, we could change the parameter $\rho$ of a variant of our algorithm
with bounded cost per iteration (see Section~\ref{sec:ski-bounded} below)
in the beginning of each idle period based on any
algorithm for the {\em experts problem} with
{\em vanishing regret}, see \citep{BianchiL06} for reference.
Under a mild assumption that the average length of the idle periods
is a positive number independent of the number of idle periods, this would also lose only
a factor $(1+\epsilon)$ compared to the best $\rho$.

Of special interest are the results on {\em shifting/dynamic} regret,
see work of \citet{Chen18}, \citet{Bianchi12}, and \citet{HallW13}.
They allow us to reach a
cost comparable not only to the algorithm with the best fixed $\rho$, but
also to the best strategy of switching between multiple values of $\rho$ a
bounded number of times.
For example, the additive regret of Fixed Share Forecaster
\citep{HerbsterW95}
with respect to a strategy with $m$ switches, is at most
\[\textstyle O\big(\sqrt{m T \ln \frac1{\epsilon_2} \ln\frac{m}{T}}\big),\]
when $T$ is the number of idle periods.

This kind of result can be very useful in scenarios where
well-predictable parts of the input are interleaved with unpredictable or
adversarial sequences which make the total prediction error high, forcing any
static strategy to use large $\rho$. The shifting regret approach allows
us to make use of the good predictions outside of the adversarial parts of the
sequence.

\subsubsection{Bounded cost per iteration}
\label{sec:ski-bounded}

The cost of Algorithm~\ref{alg:DPM-1} in one iteration may be arbitrary high
(in case of an infinite prediction error),
making it hard to use with some of the online learning techniques mentioned
above.
However, we can make sure that the cost of Algorithm~\ref{alg:DPM-1} is bounded
in each iteration at a small cost.
In fact, it is enough to ensure bounded cost for the ski rental algorithm
from Section~\ref{sec:ski} used in the construction of Algorithm~\ref{alg:DPM-1}. Note that if $P_\infty>0$, then our algorithm might never buy skis, thus leading to unbounded cost.
We can achieve bounded cost on the expense of worsening $\mu(\rho)$ by a factor of
$(1+\epsilon)$ in the following way.
For $t < u = \frac{\beta}{\alpha} (3 + 1/\epsilon)$,
we define $F_{\tau}(t)$ as for our ski rental algorithm from
Section~\ref{sec:ski} (for general $\alpha$ and $\beta$),
and $F_{\tau}(t) = 1$ for any $t \geq u$, making sure no more costs
are incurred in times $t > u$.

\begin{claim}
The algorithm with modified $F_{\tau}$ is
$\big(\rho,\, (1+\epsilon)\mu(\rho)\big)$-competitive and its cost
is never larger than $\beta(4+1/\epsilon)$.
\end{claim}
\begin{proof}
The bound on the cost is easy to see: In the worst case,
we buy at time $u$, paying
$\alpha \cdot \frac\beta\alpha (3+1/\epsilon)$ for renting
until $u$ and additional $\beta$ for buying.

We will show that the algorithm is
$\big(\rho,\,(1+\epsilon)\mu(\rho)\big)$-competitive
under the assumption $\alpha=\beta=1$.
The result for general $\alpha$ and $\beta$ then follows from
Lemma~\ref{lem:unit-costs}.

\textbf{First case:} $\tau \leq 3$.
Here, $\cost(x) \leq \rho\opt+ \mu(\rho)\cdot\eta$
for any $x \leq u$ since the modified algorithm has the same behaviour as the
one in Section~\ref{sec:ski}. If $x \geq u$, we pay an additional cost of $P_\infty$ for
buying at time $u$:
\begin{align*}
\cost(x) &\leq \rho\opt+ \mu(\rho)\cdot |\tau-x| + P_\infty\\
	&\leq \rho\opt+ \mu(\rho)\cdot \big(1+\frac{1}{|\tau-x|}\big)|\tau-x|\\
	&\leq \rho\opt + \mu(\rho) (1+\epsilon)\cdot |\tau-x|,
\end{align*}
since $P_\infty \leq \mu(\rho)$, $\tau \leq 3$ and $x \geq 3+1/\epsilon$.

\textbf{Second case:} $\tau > 3$. Here, we claim that $P_\infty = 0$
and therefore $F_{\rho,\tau}(1)=1$, i.e.,
no modification is needed.
Note that $P_\infty =0$ whenever \eqref{eq:T_UB} is tight, i.e.,
\begin{equation}
\label{eq:ap-T-UB}
e^{\tau-1} \geq \frac{\rho-\mu}{\rho - \mu + \tau\mu -1}.
\end{equation}
If we have $\mu > \mu(\tilde\rho) \approx 0.36$, then
$\tau\mu-1>0$ for any $\tau >3$, making the right-hand side
of \eqref{eq:ap-T-UB} smaller than 1.
Otherwise,
$\mu(\rho) = \frac{1-\rho\frac{e-1}{e}}{\ln 2} > 1-\rho\frac{e-1}{e}$,
and we have
\[ \frac{\rho-\mu}{\rho - \mu + \tau\mu -1}
	\leq \frac{\rho}{\rho + (\tau-1)\mu - 1}
	< \frac{\rho}{\rho - \rho \frac{e-1}{e}} = e.\]
making the right-hand side of \eqref{eq:ap-T-UB} smaller than $e$,
while $e^{\tau-1}> e$.
\end{proof}

\section{Experiments}\label{sec:exp}

\newcommand{\kumards}{PSK$_4$\xspace}
\newcommand{\kumareightds}{PSK$_8$\xspace}
\newcommand{\ftp}{FTP\xspace}
\newcommand{\angelo}{ADJKR\xspace}
\newcommand{\kumar}{PSK\xspace}

In this section, we present an experimental evaluation of our algorithms compared to existing learning-augmented ski rental algorithms for the ski rental and DPM problems\footnote{We run prior ski rental algorithms on DPM using Lemma~\ref{lem:redDPMSki} and Theorem~\ref{thm:BB}, as discussed in Sections \ref{sec:DPM-1} and \ref{section:BlumBurch}.}. We use a synthetic dataset, which was introduced by Purohit et al.~\cite{PurohitSK18}, and a real-world dataset, which is based on smartphone traces from~\cite{Zhou15}. The results of our experiments\footnote{The code and datasets we used to run our experiments are available at \url{https://github.com/adampolak/dpm}.} suggest that the performance of learning-augmented algorithms indeed degrades smoothly when the error increases, providing solutions which are better, for medium errors, than naive algorithms trusting the predictions and online (predictionless) algorithms.
In the experiments, the performance of our algorithms  degrades more smoothly when the prediction error increases than previous learning-augmented online algorithms do.
This is expected, since consistency-robustness trade-offs of previous algorithms optimize the two extreme scenarios of perfect predictions and adversarially bad predictions, whereas the notion of $(\rho,\mu)$-competitiveness also captures the case of useful but imperfect predictions.

For the ski rental problem, in addition to the classical $e/(e-1)$-competitive online algorithm, we consider the following algorithms: FTP, which blindly follows the prediction (i.e., it either buys at time 0 or never); \kumar, the randomized algorithm from \cite{PurohitSK18}; and \angelo, the deterministic algorithm from \cite{Angelopoulos19}. As three algorithms -- PSK, ADJKR, and ours -- each depend on a hyperparameter, we set them in such a way that we obtain the same consistency ($\rho$ in the notation of this paper). For example, the consistency of $\rho=1.216$ corresponds to the parameter $\lambda\approx\ln(3/2)$ for \kumar, as selected in \cite{PurohitSK18}.

We consider either two power states (of respective power consumption $\{0,1\}$ and wake-up costs $\{1,0\}$, which correspond to ski rental) or four power states, whose respective power consumption are $\{1, 0.47, 0.105, 0\}$ and wake-up costs are $\{0, 0.12, 0.33, 1\}$, values corresponding to the \emph{active}, \emph{idle}, \emph{stand-by} and \emph{sleep} states of an IBM mobile hard-drive~\cite{IraniSG03}. For four power states, we convert algorithms initially designed for two power states (i.e., multi-round ski rental) using Lemma~\ref{lem:redDPMSki} and convert the resulting algorithms to their prudent variants as discussed in Section~\ref{section:prudent}.

We either run algorithms with fixed values of $\rho$, handling each idle period independently from the past ones, or we use Theorem~\ref{thm:BB} to let the algorithms adjust the value of $\rho$ across multiple idle periods.
For randomized algorithms (i.e., \kumar and ours) we use $\rho\in \{1,1.1,1.16, 1.3, 1.4, 1.5,\frac{e}{e-1}\}$; values $1$ and $\frac{e}{e-1}$ correspond to \ftp and the classical randomized $e/(e-1)$-competitive online algorithm. For the deterministic \angelo algorithm, we do not combine it with the randomized $e/(e-1)$-competitive algorithm; instead, we additionally use $\rho\in \{1.6, 1.7, 1.8, 1.9, 2\}$, where $\rho=2$ corresponds to the classical deterministic $2$-competitive online algorithm. Finally, for FTP, we combine it only with the classical randomized online algorithm.  The value of the parameter $\epsilon_1$ is set to $0.1$.

We repeated each experiment $10$ times. As the maximum standard deviation was smaller than $0.025$, we do not print error bars on the charts and only display the average result.

\paragraph{Synthetic scenario.}
For the synthetic dataset, we generate both the input data and predictions following the approach proposed by Purohit et al.~\cite{PurohitSK18}. We use two datasets, each composed of $10\,000$ independently generated durations (of the idle periods/ski seasons). For the first dataset \kumards, durations are drawn uniformly from $[0,4]$ as originally proposed in \cite{PurohitSK18}. For the second dataset \kumareightds, durations are drawn uniformly from $[0,8]$, in order to offer instances in which also the deepest power state might be chosen by the optimal offline algorithm. We feed the learning-augmented algorithms with synthetic predictions generated as follows: each prediction is equal to the exact request plus a random noise drawn from a normal distribution of mean 0 and standard deviation $\sigma$ (rounding any negative predictions to $0$). The performance is shown as the competitive ratio observed in function of $\sigma$.

\begin{figure}
  \begin{subfigure}{0.49\textwidth}
    \centering
    \includegraphics[width=\textwidth]{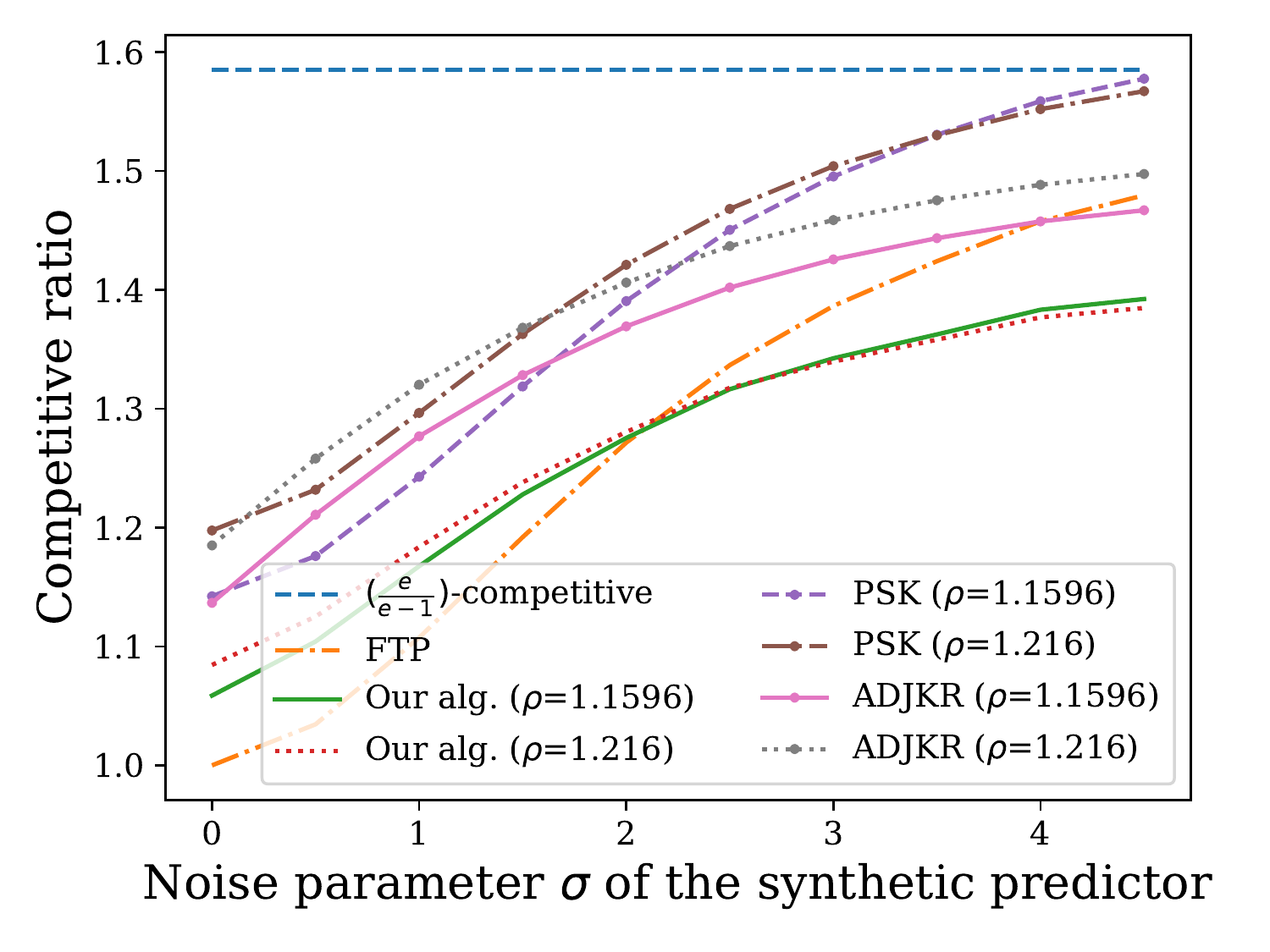}
    \caption{Algorithms with fixed consistency $\rho$.}
    \label{fig:expapp1rhos}
  \end{subfigure}
  \hfill
  \begin{subfigure}{0.49\textwidth}
    \centering
    \includegraphics[width=\textwidth]{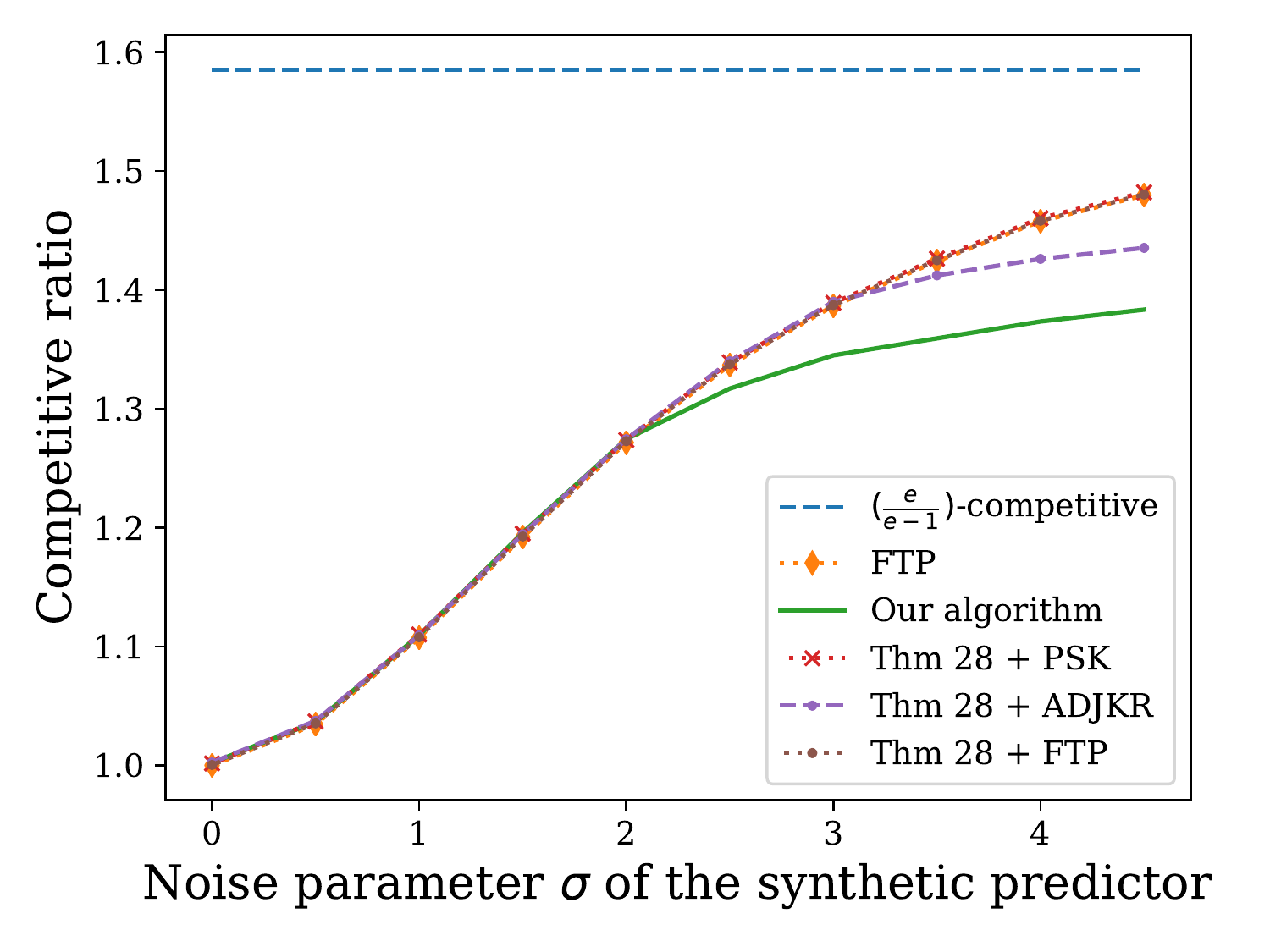}
    \caption{Combined algorithms utilising Theorem~\ref{thm:BB}.}
    \label{fig:expapp1best}
  \end{subfigure}
  \caption{Performance achieved for the two-state DPM (i.e., ski rental) problem on the synthetic dataset \kumards. Each synthetic prediction is equal to the exact request plus a random noise drawn from a normal distribution of mean 0 and standard deviation $\sigma$.}
\end{figure}

\begin{figure}
  \begin{subfigure}{0.49\textwidth}
    \centering
    \includegraphics[width=\textwidth]{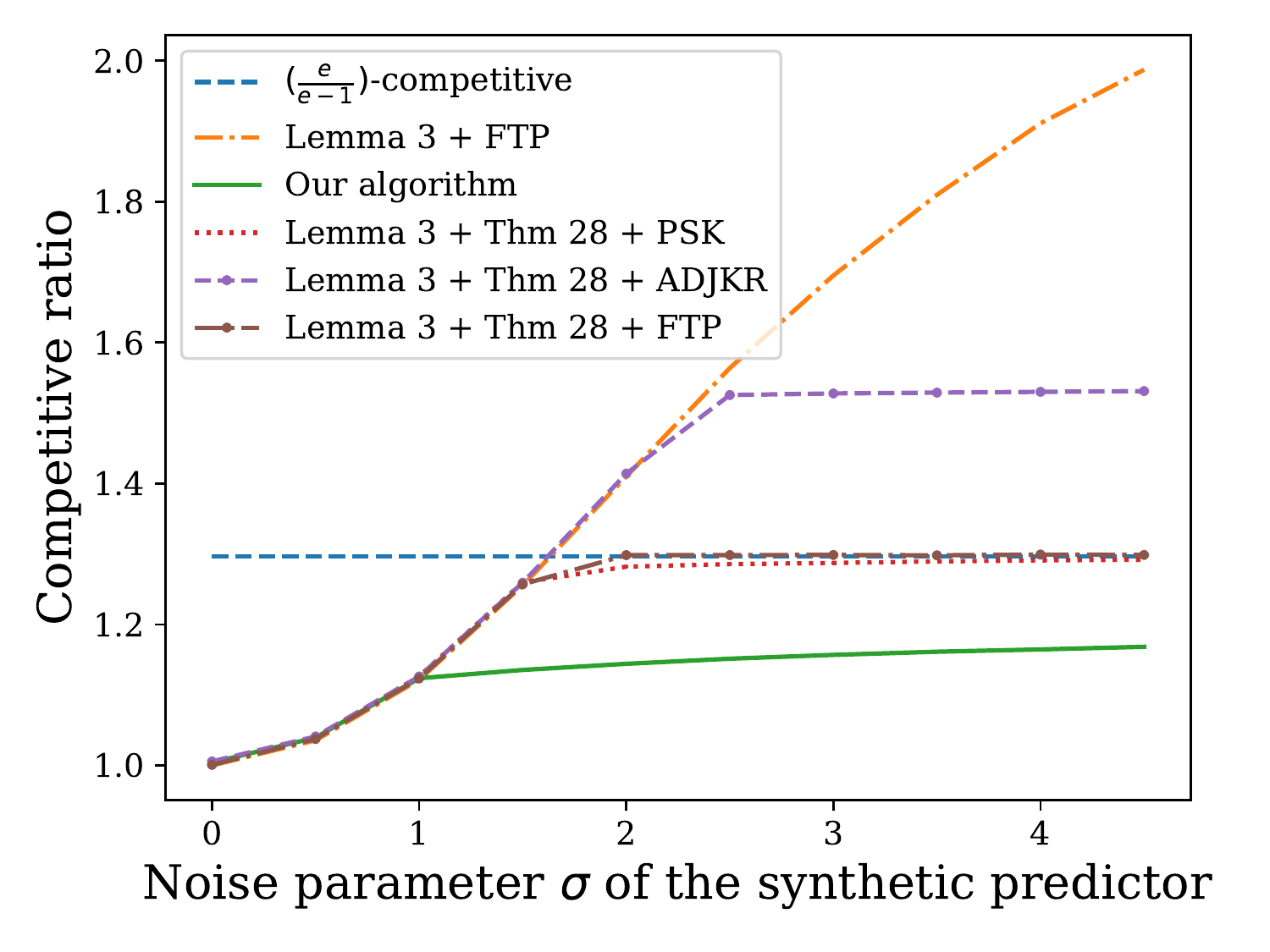}
    \caption{Dataset \kumards.}
    \label{fig:expapp2psk4}
  \end{subfigure}
  \hfill
  \begin{subfigure}{0.49\textwidth}
    \centering
    \includegraphics[width=\textwidth]{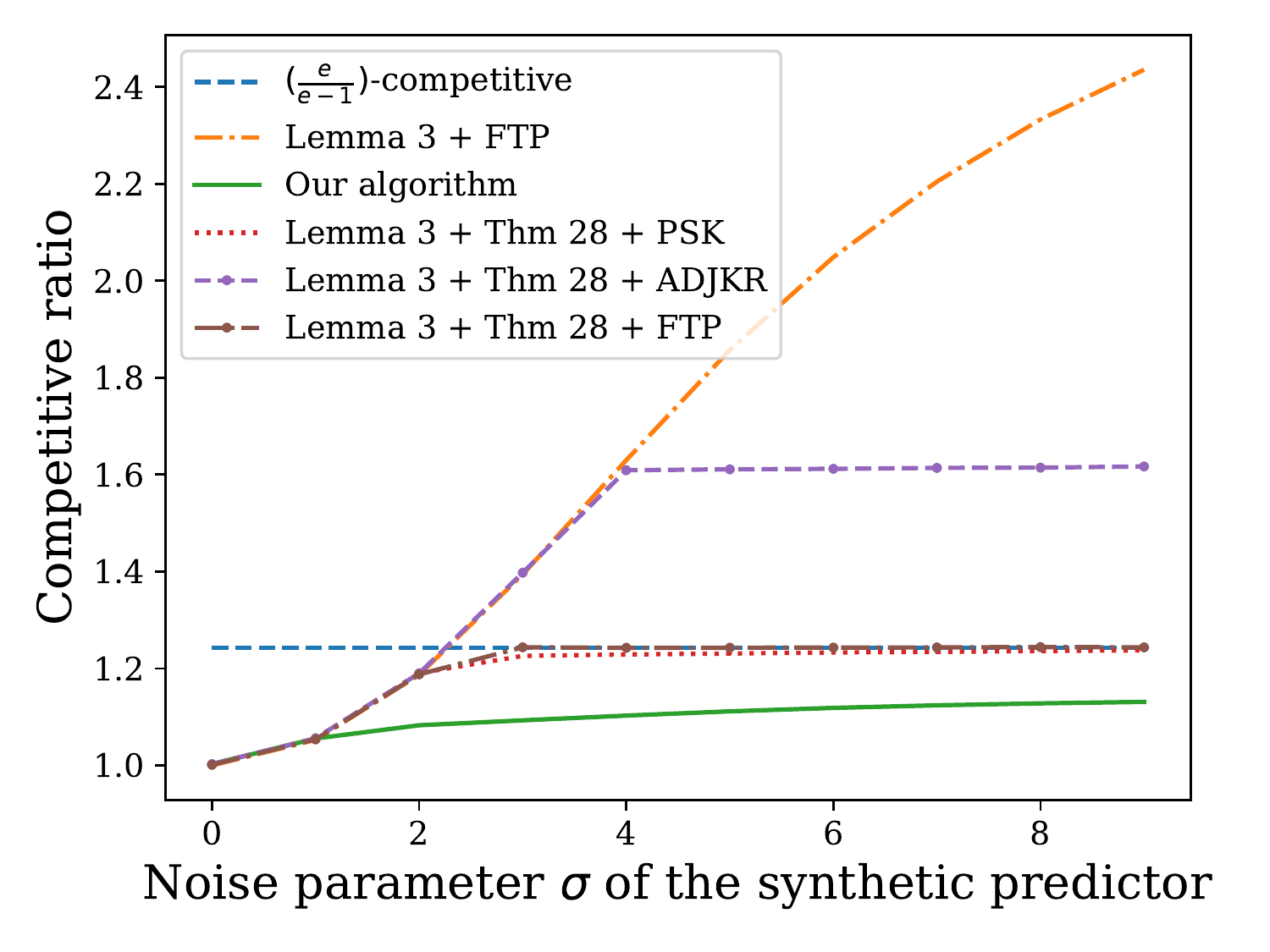}
    \caption{Dataset \kumareightds.}
    \label{fig:expapp2psk8}
  \end{subfigure}
  \caption{Performance achieved for the four-state DPM problem on the synthetic datasets by algorithms utilising Theorem~\ref{thm:BB}.}
  \label{fig:expapp2}

  \vspace{.4cm}

  \begin{minipage}[t]{0.49\textwidth}
    \includegraphics[width=\textwidth]{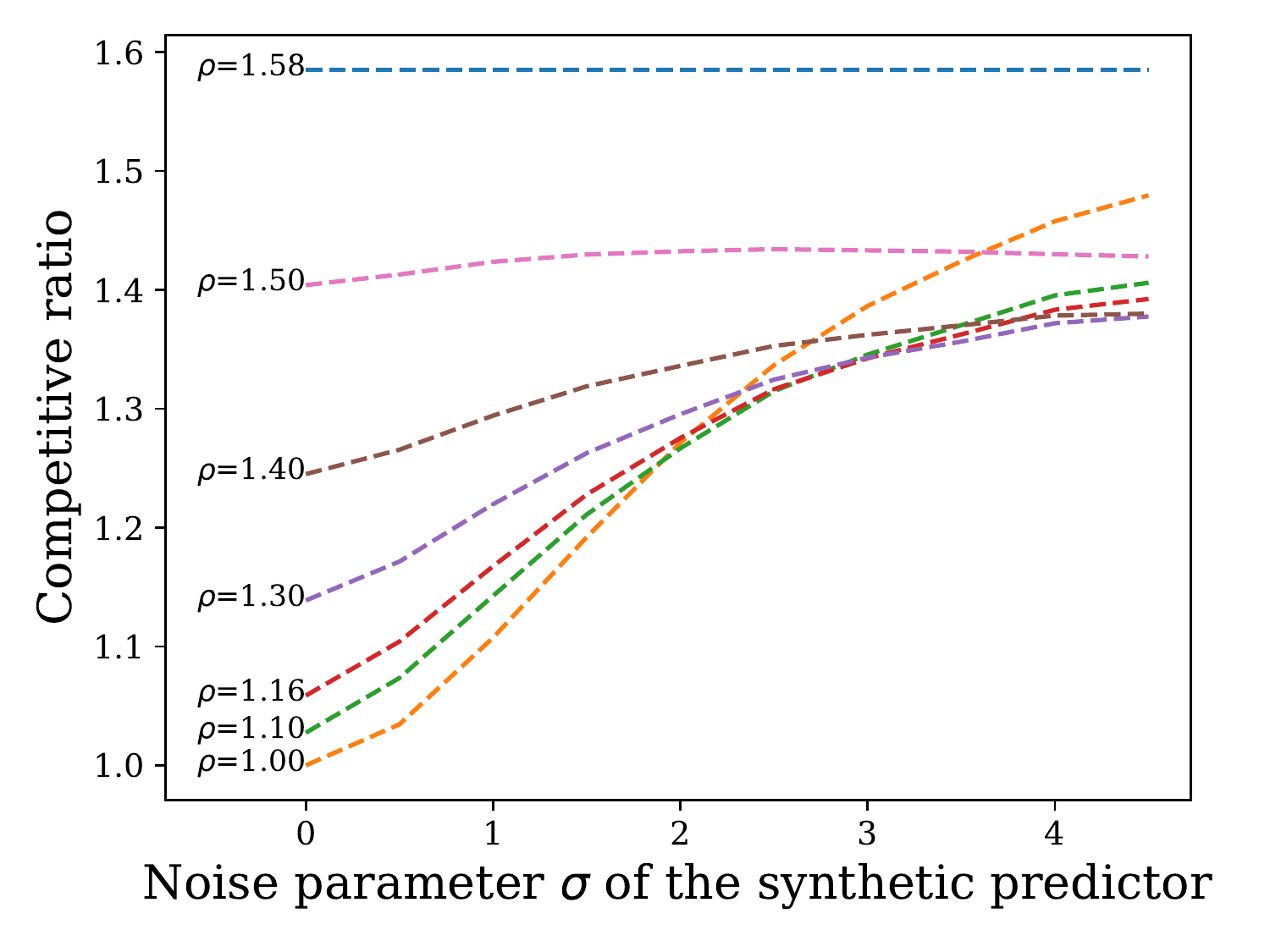}
    \caption{Performance of our ski rental algorithm for different values of $\rho$ on \kumards dataset. Recall that for $\rho = e/(e-1) \approx 1.58$, our algorithm emulates exactly the classical $e/(e-1)$-competitive online algorithm, while for $\rho=1$ it emulates the FTP algorithm.}
    \label{fig:expapp3}
  \end{minipage}
  \hfill
  \begin{minipage}[t]{0.49\textwidth}
    \includegraphics[width=\textwidth]{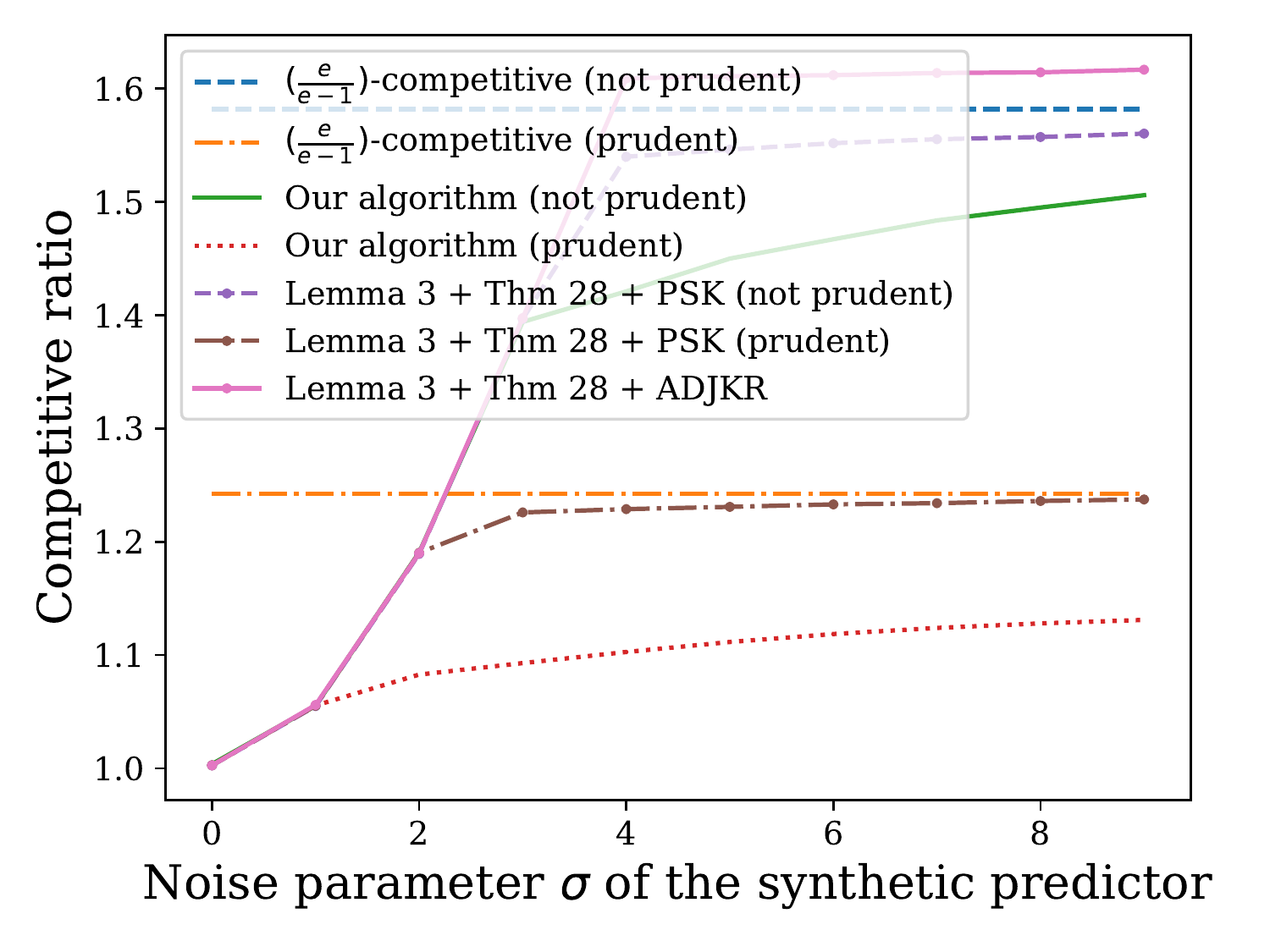}
    \caption{Impact of \emph{prudence} (see Section~\ref{section:prudent}) on performance achieved for the four-state DPM problem on \kumareightds dataset by algorithms utilising Theorem~\ref{thm:BB}.}
    \label{fig:expapp4}  
  \end{minipage}

  \vspace{-.2cm}
\end{figure}

\begin{figure}
  \centering
  \includegraphics[width=\textwidth]{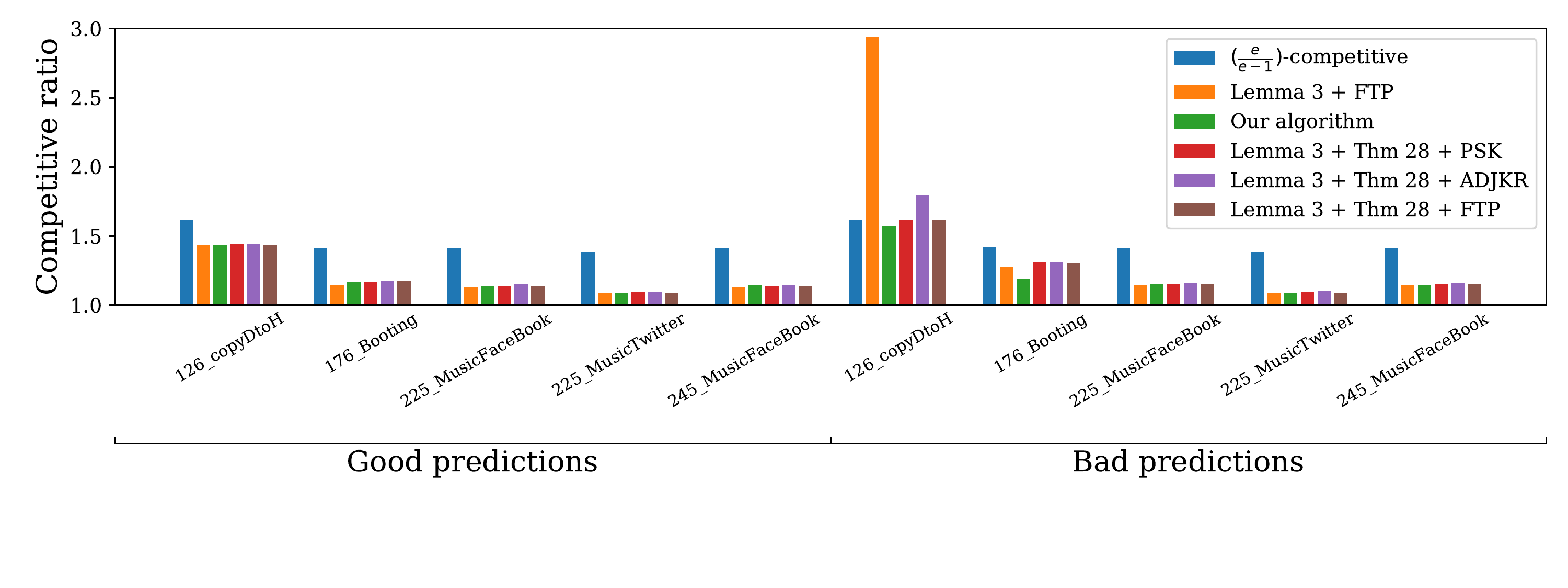}
  \vspace{-1.25cm}
  \caption{Performance achieved by algorithms utilising Theorem~\ref{thm:BB} for the four-state DPM problem on the real-world dataset based on I/O traces from Nexus 5 smartphone~\cite{Zhou15} with predictions generated by \emph{Share} algorithm of Herbster and Warmuth~\cite{HerbsterW95} as proposed by Helmbold et al.~\cite{helmbold96}.}
  \label{fig:expreal}
\end{figure}

In Figure~\ref{fig:expapp1rhos}, we plot the results for two values of the consistency $\rho$ for the two-state DPM problem (i.e., for multi-round ski rental; however the algorithms in that figure have fixed values of $\rho$ and do not carry any information between consecutive rounds).
Figure~\ref{fig:expapp1best} shows the performance of the algorithms with Theorem \ref{thm:BB} applied to automatically adjust the value of $\rho$. We can observe that for low error the algorithms perform equally well (because they quickly determine that $\rho=1$ is best, which corresponds to executing FTP). As the prediction error increases, our algorithm performs better than the other algorithms on this dataset.

Figure~\ref{fig:expapp2} presents analogous results to Figure~\ref{fig:expapp1best} for four-state DPM on both datasets, \kumards and \kumareightds. The figure shows that the algorithms PSK and ADJKR (combined with Lemma~\ref{lem:redDPMSki} and Theorem~\ref{thm:BB}) essentially perform as well as the better of the two algorithms FTP and the classical online algorithm without predictions. This reflects the fact that these algorithms were designed to optimize the trade-off between consistency and robustness in every single idle period, thus aiming primarily at the two extremes of perfect predictions and adversarially bad predictions. In contrast, the way in which our algorithm's performance degrades as the error increases is (near-)optimal in terms of $(\rho,\mu)$-competitiveness, and the algorithm's consistency and robustness are still optimal over many idle periods (but not necessarily for each idle period individually). Accordingly, our algorithm achieves a significant improvement over previous algorithms in the regime of medium-sized errors and even when predictions are only very weakly correlated with the truth.

The impact of the choice of the consistency $\rho$ on our algorithm is depicted in Figure~\ref{fig:expapp3}. Lower values of $\rho$ imply a better consistency, so a better performance when $\sigma$ is small, but may lead to worse results for larger prediction errors.

Finally, Figure~\ref{fig:expapp4} shows that converting randomized algorithms to prudent ones (see Section~\ref{section:prudent}) is crucial for achieving good experimental performance. Four-state DPM algorithms with that modification ablated perform significantly worse than their prudent variants on the \kumareightds dataset. Recall that since ADJKR is deterministic, there exists no non-prudent version of this algorithm.

\paragraph{Real-world scenario.}
Even though some of the previous works on DPM~\cite{helmbold96,IraniSG03} include real-world experiments, these papers are rather aged and the datasets used in them seem no longer available. For that reason, in order to test the algorithms in a real-world-inspired scenario, we created a dataset based on I/O traces\footnote{We obtained the traces from \url{http://iotta.snia.org/traces/block-io}.} from a Nexus 5 smartphone~\cite{Zhou15}. We took the five largest traces and extracted from them durations of idle periods between requests. Since we could not find a power states specification for that device, we used the same power states as for the synthetic experiments, i.e., an IBM mobile hard-drive's power states reported in~\cite{IraniSG03}. Because of that, we had to scale up the idle period durations so that their order of magnitude becomes similar to that in the synthetic experiments.

For this experiment, instead of resorting to synthetic predictions, we implemented a simple actual predictor, proposed by Helmbold et al.~\cite{helmbold96} in the context of spinning down disks of mobile computers. The predictor adapts the \emph{Share} learning algorithm of Herbster and Warmuth~\cite{HerbsterW95}, which is based on the multiplicative weights update method. Since it is interesting to evaluate learning-augmented algorithms both in the presence of good and bad predictions, we consider two variants of that predictor. The good variant uses hyperparameters proposed in~\cite{helmbold96}; the bad variant has the rate parameter of weight updates negated.

Figure~\ref{fig:expreal} presents the results of the experiment. In particular, on each dataset, either our algorithm performs better than all the others, or the nonrobust \ftp is the best one and all robust learning-augmented algorithms are almost equally good.

\bibliography{draft}
\bibliographystyle{abbrvnat}

\end{document}